\documentclass{article} 
\usepackage{iclr2027_conference,times}

\usepackage{amsmath,amsfonts,bm}

\def\eqref#1{equation~\ref{#1}}

\def\1{\bm{1}}

\DeclareMathAlphabet{\mathsfit}{\encodingdefault}{\sfdefault}{m}{sl}
\SetMathAlphabet{\mathsfit}{bold}{\encodingdefault}{\sfdefault}{bx}{n}

\newcommand{\E}{\mathbb{E}}

\usepackage{hyperref}
\usepackage{url}
\usepackage{bbm}
\usepackage[utf8]{inputenc} 
\usepackage[T1]{fontenc}    
\usepackage{hyperref}       
\usepackage{url}            
\usepackage{booktabs}       
\usepackage{amsfonts, amsmath, amsthm, amssymb}       
\usepackage{nicefrac}       
\usepackage{microtype}      
\usepackage{xcolor}         
\usepackage{algorithm}
\usepackage{graphicx}
\usepackage{algpseudocode}
\usepackage{array}
\usepackage{subcaption}

\usepackage{multirow}
\theoremstyle{plain}
\newtheorem{theorem}{Theorem}[section]
\newtheorem{proposition}[theorem]{Proposition}

\newtheorem{claim}[theorem]{Claim}
\newtheorem{corollary}[theorem]{Corollary}
\theoremstyle{definition}

\theoremstyle{remark}
\newtheorem{remark}[theorem]{Remark}
\renewcommand{\E}{\mathbb{E}}

\usepackage[most]{tcolorbox}
\newtcolorbox{promptbox}[1]{
    colback=gray!5!white,    
    colframe=blue!75!black, 
    fonttitle=\bfseries,    
    title=#1,               
    arc=2mm,                
    outer arc=1mm,
    enhanced,
    attach boxed title to top left={yshift=-2mm, xshift=2mm},
    boxed title style={colback=blue!75!black}
}

\title{Hiding in Plain Sight: Detectability-Aware Antidistillation of Reasoning Models}

\author{Max Hartman*\\
University of Illinois Urbana-Champaign\\
\texttt{maxh3@illinois.edu} \\
\And
Vidhata Jayaraman* \\
University of Illinois Urbana-Champaign \\
\texttt{vidhata2@illinois.edu} \\
\AND
Moulik Choraria* \\
University of Illinois Urbana-Champaign \\
\texttt{moulikc2@illinois.edu}\\
\And
Lav Varshney \\
Stony Brook University \\
\texttt{lav.varshney@stonybrook.edu}\\
}

\iclrfinalcopy 
\begin{document}

\maketitle
\lhead{Preprint.}

\begin{abstract}
Distillation via reasoning traces exposes frontier models to adversarial third parties, who can bypass their guardrails and misappropriate their capabilities. Existing antidistillation defenses perturb the teacher's sampling distribution to hinder downstream student learning, but these perturbations also degrade the trace quality even when task accuracy is preserved, often leaving visible artifacts such as syntactic incoherencies and repetition chains. Beyond eroding user trust, such traces are statistically distinguishable, allowing adversaries to filter or resample them and weaken the defense. An effective defense must therefore limit how \emph{visibly} it degrades the teacher. To address this gap, we formulate antidistillation as a max-min game whose constraint set explicitly encodes defense detectability. Satisfying both axes--- degrading the student while remaining undetectable---may require more clever perturbations, akin to sparse but targeted adversarial attacks in robustness literature. We apply the same analogy to reasoning steps, targeting \emph{thought anchors}: sentences identified in prior interpretability work as strongly influencing a model's output. We propose Thought Anchor RemovalS (TARS), which identifies and removes thought anchors, whose content is semantically redundant, so that the reasoning trace remains legible to the user. Modeling these removals as burst-deletions, we then provide formal guarantees on its detectability. Finally, we empirically demonstrate the efficacy of TARS in consistently hindering downstream distillation across teacher model scales (DeepSeek R1 Distill Qwen 7B, Gemma 3 12B), reasoning benchmarks (MMLU Pro, MATH, DeepMath), and student architectures (Llama 3.2 1B, Llama 3.2 3B, Gemma 3 1B, Qwen 3 1.7B) while preserving teacher performance. 

%
\end{abstract}

\section{Introduction}

Reasoning in large language models (LLMs) has driven substantial performance gains on multi-step tasks such as mathematics, coding, and strategic planning \citep{OpenAILearningo1, QwQ2024, AnthropicOpus4_62026}. Beyond accuracy, the traces themselves offer practical insight into how answers are produced, strengthening trust in model outputs. However, releasing traces carries a real risk: although knowledge distillation from reasoning traces was originally proposed as a route to more efficient training \citep{HSIEHLY2023, MUKHERJEEMJ2023}, distilled models have since been shown to match, and at times surpass, their teachers at a fraction of the cost \citep{GUOYZS2025}, raising pressing concerns about intellectual property for frontier laboratories competing at the edge of learnability \citep{AnthropicDistillation2026}. In fact, these concerns now sit at the forefront of U.S. national policy, with recent legislative proposals explicitly targeting the extraction of capabilities from closed-source AI models \citep{HR8283, NSTM4}, which essentially circumvent export controls on compute. Furthermore, from a safety perspective, distilled models often lose the safety alignment trained into the teacher model \citep{JAHANS2025, SHIWOW2026}, thus creating a direct pathway from publicly exposed reasoning traces to unmoderated and dangerous models. 

Defending against misuse therefore requires traces that remain useful to users while being ineffective for training a student, a goal known as \emph{antidistillation}. Summarization is a commonly deployed preventive measure, but compressing reasoning traces can introduce hallucinations \citep{TANGSWB2024, SONGSSC2024} and obscure key logical steps, eroding user trust. Pioneering work in this area \citep{SAVANITFX2025, LITZQ2025} instead alters the teacher's output distribution, which requires either access to proxy students during trace generation or fine-tuning the teacher itself. Doing so, however, leaves visible artifacts in the output traces. Specifically, Antidistillation Sampling (ADS) \citep{SAVANITFX2025} often makes poisoned traces grammatically incoherent, inserts nonsensical tokens, and falls into loops of repeated tokens, while Defensive Output Generation (DOGe) \citep{LITZQ2025} renders traces difficult for users to follow by abruptly changing topic or repeating itself in different words (Figure~\ref{fig:shortcomings}). Such \textit{detectability}, beyond eroding user trust, makes poisoned traces statistically distinguishable from clean ones. As a result, both methods are vulnerable to simple countermeasures such as filtering, resampling, or prompting a smaller LLM to flag poisoned traces, making them easy to bypass.

\begin{figure}[htbp]
     \centering
     \begin{subfigure}[b]{0.45\textwidth}
         \centering
         \includegraphics[width=\textwidth]{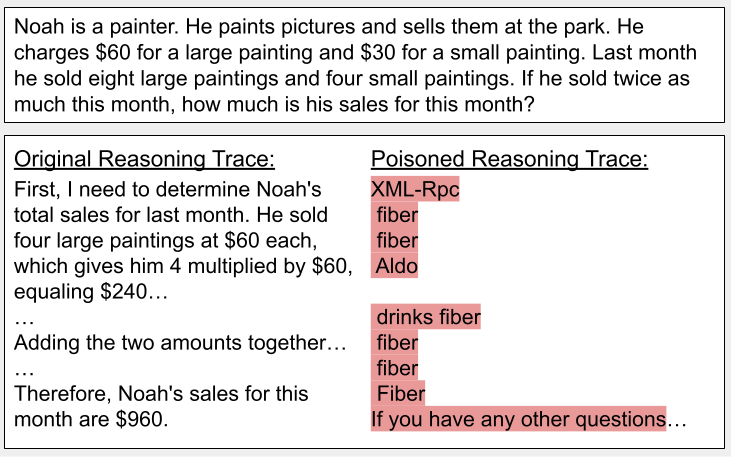}
         \caption{ADS example student output on a prompt. Grammatical incoherency (such as "fiber fiber Aldo drinks fiber...") throughout the poisoned trace (highlighted in red).}
         \label{fig:top_left}
     \end{subfigure}
     \hfill
     \begin{subfigure}[b]{0.45\textwidth}
         \centering
         \includegraphics[width=\textwidth]{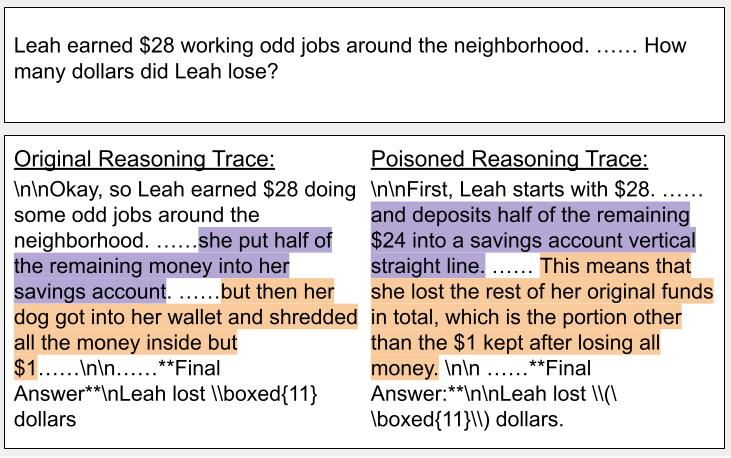}
         \caption{DOGe example student output on a prompt. Poisoned trace looks plausible, but rewritten parts (such as "savings account vertical straight line") are hard to understand.}
         \label{fig:top_right}
     \end{subfigure}
     
     \caption{Original and Poisoned Reasoning traces by ADS and DOGe show the shortcomings of existing defense methods on samples from the GSM8K \citep{COBBEKBC2021} dataset. These methods are all detectable or obscure the trace to the degree that the trace may be unhelpful for an actual user.}
     \label{fig:shortcomings}
\end{figure}
\vspace{-0.2cm}

To address current limitations, we formulate antidistillation in a game-theoretic framework which explicitly allows for the enforcement of detectability constraints, drawing on its connection to data poisoning. We then draw on recent work in LLM interpretability, which has shown that specific sentences, known as \textit{Thought Anchors} \citep{BOGDANMNC2025}, in reasoning traces are disproportionately important to the model output and reasoning flow to provide a method which accounts for detectability. We argue that poisoning Thought Anchors can disrupt the downstream student's ability to distill the teacher's reasoning. Empirically, we show that students trained on poisoned traces exhibit degraded distillation performance relative to those trained on unpoisoned traces. We further characterize the detectability of our method both theoretically and empirically, and compare it with ADS using appropriate detectability metrics, demonstrating competitive performance. Our contributions are as follows: 
\begin{enumerate}
    \item \textbf{Theoretical Formulation:} We formulate antidistillation as a max-min optimization problem, mapping data poisoning frameworks to LLM distillation setting to ground future research.
    \item \textbf{Thought Anchor RemovalS (TARS) Method:} We introduce a method that implements targetted poisoning using thought anchors. Across multiple datasets and choice of teacher-student model pairs, we comprehensively show that it degrades downstream distillation while retaining coherence, unlike competing methods.
    \item \textbf{Detectability Bounds:} We provide rigorous bounds on the detectability of the TARS defense by modeling the method as performing $k$ burst deletions.

\end{enumerate}

\section{Related Work}
\textbf{Antidistilation}: Antidistillation techniques aim to prevent adversarial model distillation attacks by modifying a teacher's outputs to be poor training data for the student. \cite{SAVANITFX2025} proposed ADS, which perturbs the next-token probability distribution of the teacher model by encouraging the generation of tokens to degrade a student model, via a proxy model. \cite{MAYZV2026} proposed a prompting approach, which instructs the LLM to generate less distillation-friendly reasoning traces, and a gradient-based approach, which perturbs tokens post-generation to maximize test loss on a proxy model. DOGe \citep{LITZQ2025} finetunes the final layer of the teacher model with a loss to prevent student learning. \cite{DINGCD2025} finetunes an LLM to remove self-talk tokens from the reasoning trace. There is also concurrent work, such as SGRE \citep{LiPXWCZ2026}, which edits the structure and verbalization of completed traces to preserve answer accuracy and naturalness, LADS \citep{DiaoGAZHH2026} correlates responses across distillation accounts, making semantically similar prompts across accounts have shared randomness while allowing individual users to see no disruption, and \citep{AllouahHKS2026}, who formulate a minimax game with an adaptive student and derive a product-of-experts decoding defense. Complementary evaluations find that defense efficacy varies by task \citep{Jiang2026} and by the attacker's query budget, data budget, and API access \citep{LiborPAPT2026}. Building on this literature, we explicitly incorporate detectability into a game-theoretic formulation and target sparse, influential portions of reasoning traces to impede distillation while maintaining trace coherence.

\textbf{Reasoning in LLMs}:
Chain-of-thought (CoT) reasoning, in which models are trained to generate intermediate reasoning steps, emerged in recent years as a technique to improve the performance of models on complex multi-step tasks \citep{WEIWS2022}. The Tree-of-thought (ToT) technique generalized and improved CoT by enabling LLMs to consider different reasoning paths \citep{YAOYZS2023}. Reinforcement learning post-training with human-annotated reasoning data was used as an approach to distill reasoning capabilities models. Reinforcement learning without human-labeled data was then proposed \citep{OpenAILearningo1, QwQ2024}. Recently, models have been trained to learn reasoning by distilling knowledge from a large reasoning model, such as the R1-distill series \citep{GUOYZS2025}. Interpretability of Reasoning models has recently been studied, such as work which identifies how parts of a reasoning trace at the sentence and token-level are unequally important to answer generation \citep{WANGYGZ2025, BOGDANMNC2025}.

\textbf{Data Poisoning}: 
A useful motivation for a framework for antidistillation comes from data poisoning. Data poisoning is an integrity attack on the training phase of machine learning models which manipulates the training data distribution to harm the model \citep{FANYLQ2022}. The adversary is additionally constrained to perturbations that remain close to legitimate data, so that the attack is not trivially detectable. Early work considered attacks on support vector machines and neural networks \citep{BIGGIONL2012, SHAFANIH2018}. These attacks have also been shown to be effective against LLMs \citep{CARLINIJC2024, SOULYRC2025}. Further work has studied targeted, less-detectable poisoning methods \citep{suVS2017,wiyatnoX2018}. Our work shows that antidistillation can be formulated as a generalization of the classical game-theoretic data poisoning framework.

\section{Formulation}
To motivate our proposed formulation for antidistillation, we first look at the classical data poisoning formulation. Herein, an attacker attempts to find a set of data points sufficiently ``close" to real data points that maximally harm the performance of a given statistical model. This format can be viewed as a Stackelberg game, where the attacker commits to a poisoned dataset $\tilde{D}$, after which the defender attempts to find the optimal model parameters from it. Observe that the best action for the defender in this game is:
\[
\theta^*(\tilde{D}) := \arg\min_{\theta \in \Theta} \tilde{\mathcal{L}}(\theta; \tilde{D} \cup D),
\]
while the optimal action for the attacker to take is: 
\[
\tilde{D}^* := \arg\max_{\tilde{D}\in \mathcal{B}} \mathcal{L}(\theta^*(\tilde{D})).
\]
Therefore, the game can be viewed as the following bi-level optimization problem: 
\begin{equation}\label{eq:data_pois}
    \tilde{D} = \arg \max_{\tilde{D} \in \mathcal{B}} \mathcal{L}(\arg\min_{\theta} \tilde{\mathcal{L}}(\theta; \tilde{D}\cup D)),
\end{equation}
where $\tilde{D}$ is the poisoned dataset, $\mathcal{B}$ is the set of admissible poisoned datasets, $D$ is a clean dataset also given for training, $\mathcal{L}$ is the loss function to evaluate the optimal model $\theta$, and $\tilde{\mathcal{L}}$ is the loss function for evaluating the model on the poisoned data \citep{FANYLQ2022}. The solution, $(\tilde{D}^*, \theta^*(\tilde{D^*}))$, to~\eqref{eq:data_pois} constitutes the Stackelberg equilibrium for this data poisoning game. Classically, the admissible set $\mathcal{B}$ may include constraints like
\begin{equation}\label{eq:dp_norm}
\mathcal{B} = \mathcal{B}_\epsilon^p = \{(x+\delta, y) \mid \|\delta\|_p \leq \epsilon\},
\end{equation}
or
\begin{equation}\label{eq:dp_detect}
\mathcal{B} = \mathcal{B}_\tau = \{(\tilde{x},y) | S(x) < \tau\},
\end{equation}
where $S(\cdot)$ is some known detector that the attacker assumes the defender will use and thus wants to prevent from detecting the poisoning. These constraints are explicitly created to prevent detection of the poisoning as, particularly in the case where poisoned data is simply appended to clean data, if a defender can detect the poisoning they can simply avoid training on that data. 

While this framework models something very similar to antidistillation, we argue that there are 2 key changes that need to be made. First, in antidistillation, the model provider will not provide clean traces for training so modeling the student's available data as $D \cup \tilde{D}$ it not a good assumption. Secondly, this framework only considers a student who does empirical risk minimization (ERM) on one specific model class over the data but in practice distillers will do comprehensive post-processing on the received data and the model provider may not know the specific model class being used by the student. As such, we modify Equation~\ref{eq:data_pois} to 
\begin{align}\label{eq:antidistill}
    \tilde{D}_T^* = 
    \arg\max_{\tilde{D}_T \in \mathcal{B}} \inf_{\mathcal{A} \in \mathcal{F}} \E_{x\sim p_{\mathrm{data}}}\left[ \mathcal{L}\left(\mathcal{A}\left(\tilde{D}_T\right) ;x \right)\right],
\end{align}
where $\mathcal{B}$ is the set of admissible poisoned datasets which can also be dependent on the teacher model (e.g. to allow for retraining), $\mathcal{F}$ is the set of admissible set of training algorithms $\mathcal{A}: \mathcal{X}^n \to \mathcal{H}$ which take the poisoned data and output a hypothesis, $p_{\mathrm{data}}$ is the true data distribution, and $\mathcal{L}$ is the population loss. This framework encompasses the data poisoning framework by taking $\mathcal{A}$ to be ERM and by taking $\mathcal{B} = \{(D_c \cup \tilde{D}) : \tilde{D} \in \mathcal{B}_{\mathrm{dp}}\}$.

Similar to how detectability constraints are enforced on $\mathcal{B}$ in data poisoning, one can take enforce analogous constraints in the case of antidistillation. For example, we specifically consider the case of $\mathcal{B}$ being a KL-divergence ball around the original teacher distribution (analogous to Equation~\ref{eq:dp_norm}) but one can also consider datasets which satisfy some given detector such as a perplexity or accuracy detector (analogous to Equation~\ref{eq:dp_detect}).

To connect this formulation to existing works in literature, we show in Appendix~\ref{appendix:gen_prior_methods} how this framework subsumes ADS \citep{SAVANITFX2025} and DOGe \citep{LITZQ2025}. Crucially, both of these methods make implicit assumptions on student model architecture by way of a suitable proxy, while being easy to detect. Therefore, the main takeaway here is that \emph{any good antidistillation defense in practice should limit detectability while also be downstream student-agnostic to a realistic degree, as the model provider often has no knowledge of the student's model/training procedure.}

\section{Method}
\label{sec:method} 

\subsection{Understanding units of Reasoning Structure}\label{sec:units}

A major limitation of prior works such as ADS and DOGe, especially in the context of practical detectability, is the need to perturb the teacher's output distribution at every token position. This is analogous to dense $\ell_2/\ell_{\infty}$ attacks on images \citep{dezfooliFF2016, madryMSTV2018} in the adversarial robustness literature, and the main reason behind noticeable artifacts as shown in Fig. \ref{fig:shortcomings}. Extending the same analogy, a well-established finding is that sparse, targeted perturbations can often be equally as effective while being harder to detect by humans. Specifically, one-pixel attacks \citep{suVS2017} and Jacobian-based Saliency Map Approach (JSMA) \citep{wiyatnoX2018} demonstrate that perturbing a small number of carefully chosen input dimensions suffices to fool classifiers, often with greater success than uniform perturbation. Motivated by this analogy, we question whether it is possible to transfer this principle to LLM reasoning. 

As pixels for a natural image, tokens serve as the fundamental units of language data. A natural question to then ask is if there are specific tokens in a trace that are critical to the reasoning process. Prior work has shown that particular tokens can exert disproportionate influence on planning and downstream outputs \citep{ANTHROPICBIOLOGY2025}, or induce behaviors such as backtracking and correction \citep{muennighoffYSLFHZLCH2025}. However, systematic token-level interventions are hard because such signals can be noisy \citep{minLWZZZS2026}, making them difficult to identify and reliably target and may also induce noticeable artifacts. Instead, a recent line of work treats sentences as the basic units of analysis. \cite{BOGDANMNC2025} identifies a small subset of sentences, termed \textit{thought anchors}, that exert disproportionate counterfactual influence on the remainder of a reasoning trace and its final response. Besides planning sentences, such anchors are often uncertainty-management (e.g., ``Wait, I made a mistake''), or self-reflection (e.g., ``Maybe I can calculate this in decimal'') statements, while many active-computation steps simply execute operations determined by the preceding reasoning. Concurrent work \citep{kapoorGCBHK2026} similarly finds that errors at particular steps in multi-step reasoning can cascade through the remainder of a solution, while perturbations to routine steps have much less effect. Together, they suggest that sentence-level steps provide a more structured and reliable intervention target than individual tokens and thought anchors serve as our point of inquiry.


\subsection{Identifying Thought Anchors}

We implement thought anchors through their attention-based \emph{broadcasting} behavior, described in \cite{BOGDANMNC2025}. A key characteristic is that a small set of \emph{receiver heads} stand out: rather than spreading attention over the entire trace, they consistently attend back to a handful of sentences, often ones that appeared much earlier in the trace. We identify these heads by measuring how concentrated the attention each head directs at earlier sentences is, using its kurtosis. A head that repeatedly singles out a few sentences has a sharply peaked distribution and high kurtosis, whereas a head that attends broadly has low kurtosis. Each sentence is assigned a \emph{broadcasting score}: the attention it receives through these receiver heads from sentences several steps later in the trace. High-scoring sentences are treated as thought anchors. We provide details of receiver-head calibration and sentence scoring in Algorithm~\ref{algo:detectTA} of Appendix~\ref{appendix:TARS}.

Importantly, structural influence and semantic role are distinct. The broadcasting score identifies which sentences are influential, while a separate reasoning-type classification distinguishes planning, uncertainty management, self-checking, active computation, fact retrieval, and other steps. We find that many influential sentences are planning or self-reflective in nature: they shape how reasoning proceeds without directly adding problem-specific computations or facts. This makes them particularly attractive targets, as intervening on them can potentially disrupt the reasoning structure available to a student while largely preserving the substantive content and final answer of the trace. An illustrative example is presented in Fig. \ref{fig:though_anchors_illustrate}, where we used Gemini 3 as a judge to label anchors.

\begin{figure}[h]
    \centering
    \begin{subfigure}[t]{0.45\linewidth}
        \centering
        \includegraphics[width=\linewidth]{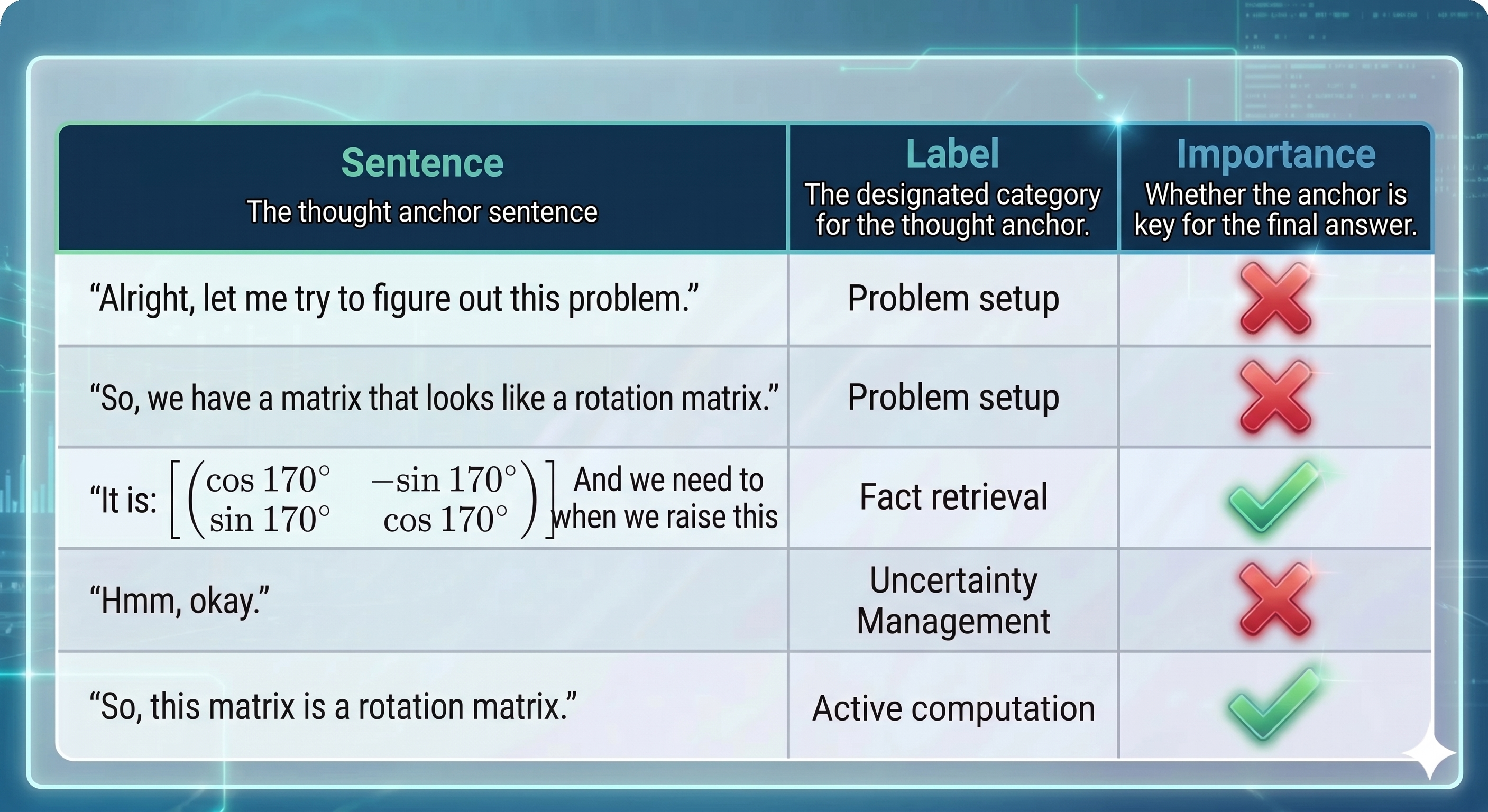}
        \caption{Thought Anchor labeling example in MATH.}
        \label{fig:th_math}
    \end{subfigure}
    \hfill
    \begin{subfigure}[t]{0.45\linewidth}
        \centering
        \includegraphics[width=\linewidth]{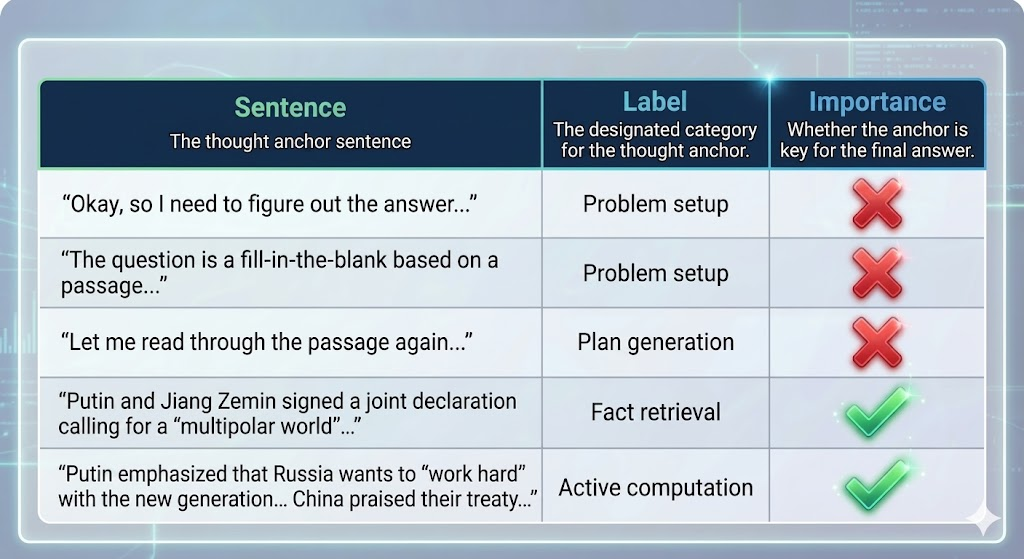}
        \caption{Thought Anchor labeling example in MMLU.}
        \label{fig:th_mmlu}
    \end{subfigure}
    \caption{Illustration of Thought Anchors and sentence types of reasoning traces generated by Deepseek R1 Distil Qwen 7B. We use the same taxonomy for determining thought anchors as \cite{BOGDANMNC2025} and ask Gemini 3 to label the thought anchors according to this scheme.}
    \label{fig:though_anchors_illustrate}
\end{figure}

\subsection{A Proof of Concept: Thought Anchors Transfer to Student}\label{sec:anchor_transfer}

While thought anchors may influence the teacher's own reasoning, that does not imply that a student distilled from the teacher's traces inherits similar reasoning mechanisms. Therefore, as a first sanity check, we generate reasoning traces on MATH and MMLU with a teacher model (DeepSeek-R1-Distill-Qwen-7B) and fine-tune a student model (Llama-3.2-3B) on them. If the student inherits the teacher's anchors, the two models should concentrate attention on the same sentences of the trace.

To do this, we isolate the thought anchor receiver heads for each model (see Algorithm~\ref{algo:detectTA} for details) and dataset separately. Both models read 128 teacher traces, and both score every sentence by the attention it receives through its own receiver heads. Once each trace is scored, we measure the similarity of how the two models rank traces sentences (Spearman $\rho$) and how much their five highest-scoring sentences overlap (top-5 Jaccard similarity), averaged over traces. We refer to these measures collectively as \emph{anchor agreement}. The experimental details and parameters are in Appendix~\ref{appendix:thought_anchor}

\begin{table}[th]
\centering
\caption{Anchor agreement between the teacher (DeepSeek-R1-Distill-Qwen-7B) and the student (Llama-3.2-3B) before and after distillation, averaged over 128 traces. Spearman $\rho$ measures how similarly the two models rank all sentences, and top-5 Jaccard measures the overlap between their five highest-scored sentences. $\Delta$ denotes the relative change after distillation.}
\label{tab:anchor_transfer}
\small
\begin{tabular}{lcccc}
\toprule
 & \multicolumn{2}{c}{MATH} & \multicolumn{2}{c}{MMLU} \\
\cmidrule(lr){2-3} \cmidrule(lr){4-5}
Student & Spearman $\rho$ & Top-5 Jaccard & Spearman $\rho$ & Top-5 Jaccard  \\
\midrule
Undistilled & 0.70 & 0.33 & 0.67 & 0.23  \\
Distilled   & 0.78 & 0.41 & 0.71  & 0.31 \\
\midrule
$\Delta$    & 11\% & 24\% & 6\% & 35\%  \\
\bottomrule
\end{tabular}
\end{table}

To isolate the effect of distillation, we compare the teacher and student scores before and after fine-tuning. As Table~\ref{tab:anchor_transfer} shows, distillation increases anchor agreement on both metrics, indicating that the student inherits the teacher's thought anchors. Combined with prior work establishing their importance for reasoning \citep{BOGDANMNC2025}, this suggests that disrupting thought anchors in the teacher's traces may disrupt the reasoning a student learns from them, forming the basis of TARS.

\subsection{TARS: Thought Anchor RemovalS}
\label{sec:tars}

Having established that thought anchors transfer from teacher to student, we now introduce TARS\footnote{...an antidistillation method, not the robot (for the optimal experience we suggest listening to Hans Zimmer)}, which detects and deletes thought anchors from reasoning traces as a defense against distillation. TARS removes thought anchors in order of reasoning category and then broadcasting score, deleting first the uncertainty-management, planning, and self-checking steps, which shape how the reasoning proceeds without adding problem-specific content. Receiver heads are calibrated once per teacher, after which each trace requires a single teacher forward pass to score. Unlike ADS, TARS edits traces after generation rather than altering the teacher's sampling, and requires neither a proxy student nor gradient computation during generation. The teacher's final answers are therefore unchanged, and the deletion budget $k$ trades off distillation resistance against detectability, which can be swept at no additional cost. TARS is formalized in Algorithm~\ref{algo:deleteTA}, and the thought anchor detection routine in Algorithm~\ref{algo:detectTA} (Appendix~\ref{appendix:TARS}), together with all hyperparameters.

\begin{algorithm}[!h]
\caption{\textsc{TARS} removes up to $k$ sentences, ordered by reasoning category and score, deleting steps a student would most learn from while keeping the problem setup and final answer intact.}
\label{algo:deleteTA}
\begin{algorithmic}[1]
\Require response $y$, sentences $(s_1, \dots, s_n)$, scores $\mathrm{score}$, budgets $\mathcal{K}$
\State $c_j \gets \textsc{Classify}(s_j)$ for $j = 1, \dots, n$
\State $\mathrm{tier}(c) \gets \begin{cases}
  1 & c = \texttt{uncertainty\_management} \\
  2 & c = \texttt{plan\_generation} \\
  3 & c = \texttt{self\_checking} \\
  4 & \text{otherwise}
\end{cases}$
\State $P \gets \{\texttt{problem\_setup}, \texttt{final\_answer}\}$
\State $Q \gets \{\, j : c_j \notin P \,\}$ sorted by $(\mathrm{tier}(c_j)\uparrow,\ \mathrm{score}[j]\downarrow)$
\For{each $k \in \mathcal{K}$}
  \State $y'_k \gets y$ with sentences $Q[1 .. \min(k, |Q|)]$ removed
\EndFor
\State \Return $\{\, y'_k : k \in \mathcal{K} \,\}$
\end{algorithmic}
\end{algorithm}




\section{Experiments}
\label{sec:experiments}

We now evaluate the efficacy of TARS in disrupting downstream distillation by characterizing how the student downstream accuracy changes when distilled using traces from a teacher model with an increasing thought anchor removal budget: $k\in\{5,10,\dots, 30\}$. We consider two teacher models: DeepSeek-R1-Distill-Qwen-7B (consistent with prior work) and Gemma 3 12B. For DeepSeek-R1-Distill-Qwen-7B, we use Llama-3.2-3B,  Llama-3.2-1B, and Gemma 3 1B as student models, with MATH \citep{HENDRYCKSBK2021} and MMLU Pro \citep{WANGMZN20224} as datasets. For Gemma 3 12B, we use the stronger Qwen 3 1.7B as the student, and DeepMath \citep{HELXLCW2026} as the dataset. 

Then, to ground our results in the broader antidistillation landscape, we compare TARS to the main prior work, ADS. We plot the teacher accuracy versus student accuracy for varied intensity of attacks (varying  $k$ for TARS and $\lambda$ for ADS). We do not directly compare to DOGe due to the key issue: as part of the defense, DOGe requires adversarially retraining the final layer in the teacher model on a set of benchmarks. In practice, fine-tuning on a specific subset without having a larger, diverse mix of general data leads to catastrophic forgetting \citep{kothaSR2024, LuoYMLZZ2025, WuHJV2026}, making it practically infeasible. This is without considering the post-training and guardrails that modern models require, or the visible signs of model degradation as in Fig. \ref{fig:shortcomings}.

To compare to ADS, we jointly consider detectability and student accuracy. This framing rewards a method that suppresses student distillation and remains hard to detect. We measure detectability by evaluating Perplexity \citep{JELINEKMB1977, RADFORDNSS2018}, TF-IDF \citep{SALTONB1988}, Semantic Similarity, token change length, teacher accuracy, and an LLM Judge, measured relative to baseline. Further details on the distillation and detectability setup are in Appendices \ref{appendix:detectability_metrics} and \ref{appendix:experimental_setup}.

\subsection{Results}

Figures \ref{fig:teacher-student-accuracy-math} and \ref{fig:teacher-student-accuracy-gsk8k} show student distillation degradation for varying maximum sentence removal budgets across the experimental suite described in Section \ref{sec:experiments}. In the plots, the ideal antidistillation method lies in the bottom-left region, where student accuracy is significantly degraded with minimal teacher degradation. The inherent teacher-student accuracy trade-off in ADS does not exist in TARS. There is a positive relationship between the number of poisoned thought anchor sentences and the accuracy drop between baseline and poisoned distillation. We further verify that these accuracy results are due to thought anchors (and not just any sentence) by including an additional ablation experiment where we compare TARS to random sentence removal in Appendix \ref{appendix:random_poisoning}. We provide the additional results with the remaining student/dataset pairs in Appendix \ref{appendix:more_TARS_results}.


\begin{figure}[!t]
    \centering
    \begin{subfigure}[b]{0.49\linewidth}
        \centering
        \includegraphics[width=\linewidth]{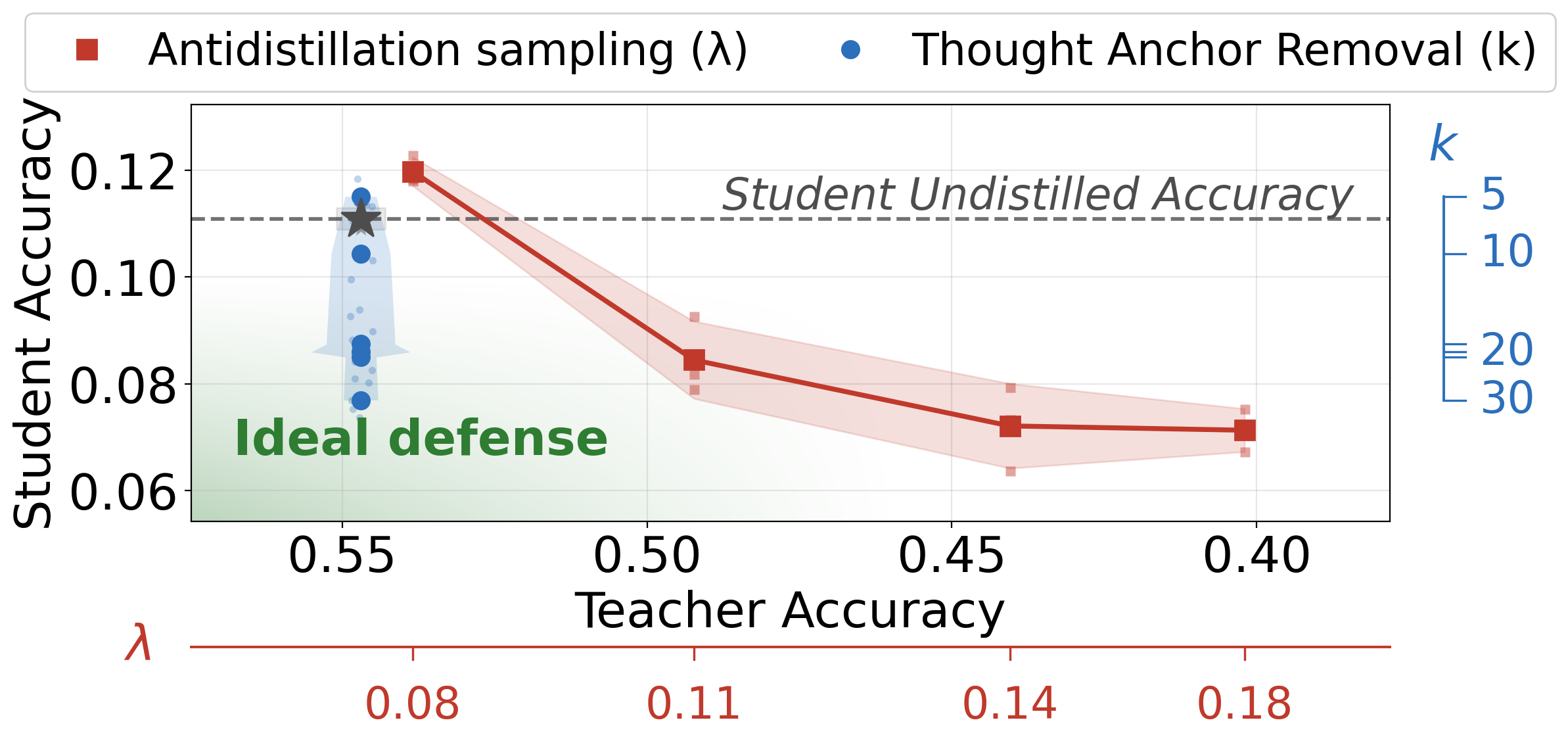}
        \label{fig:accuracy_drop_llama3_3b_math}
    \end{subfigure}
    \hfill
    \begin{subfigure}[b]{0.49\linewidth}
        \centering
        \includegraphics[width=\linewidth]{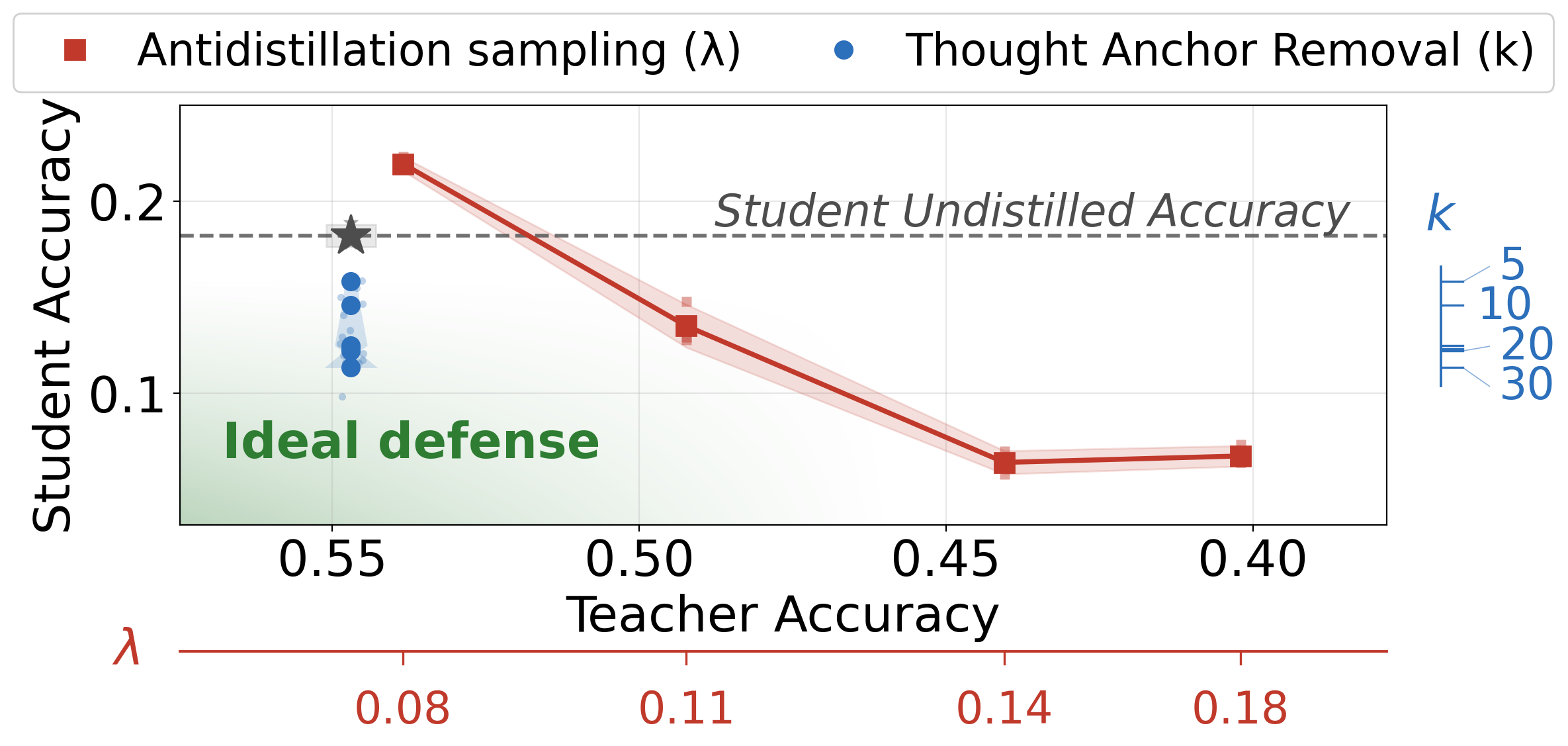}
        \label{fig:accuracy_drop_gemma3_1b_math}
    \end{subfigure}
    \caption{Student versus teacher accuracy for ADS (red) and TARS (blue) on MATH, with DeepSeek-R1-Distill-Qwen-7B as the teacher. The students used were Llama 3.2 3B (left) and Gemma 3 1B (right); Llama 3.2 1B is provided in Appendix \ref{appendix:failure_case_mmlu}.}
    \label{fig:teacher-student-accuracy-math}
\end{figure}

\begin{figure}[!t]
    \centering
    \begin{subfigure}[b]{0.49\linewidth}
        \centering
        \includegraphics[width=\linewidth]{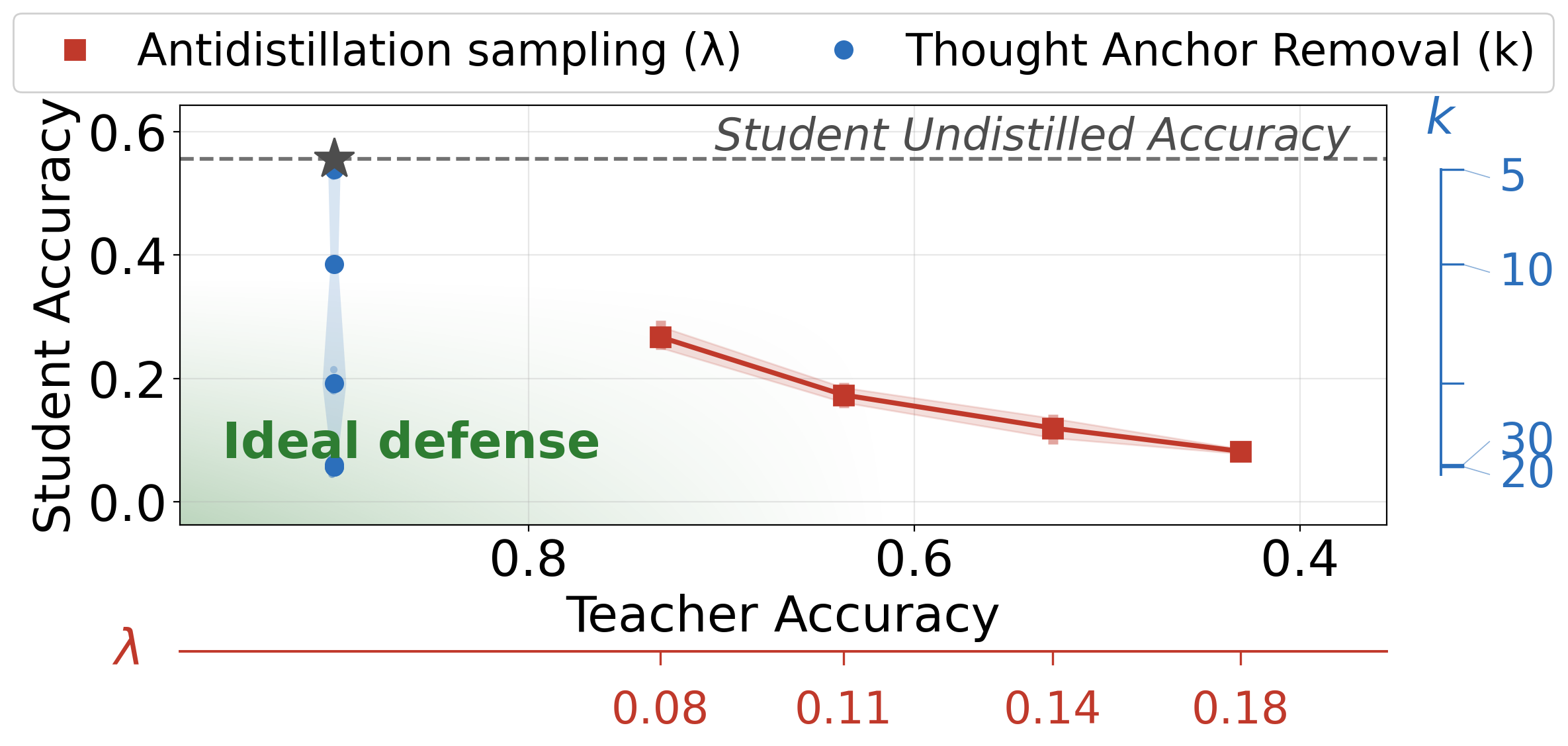}
        \label{fig:accuracy_drop_llama3_3b_gsm8k}
    \end{subfigure}
    \hfill
    \begin{subfigure}[b]{0.49\linewidth}
        \centering
        \includegraphics[width=\linewidth]{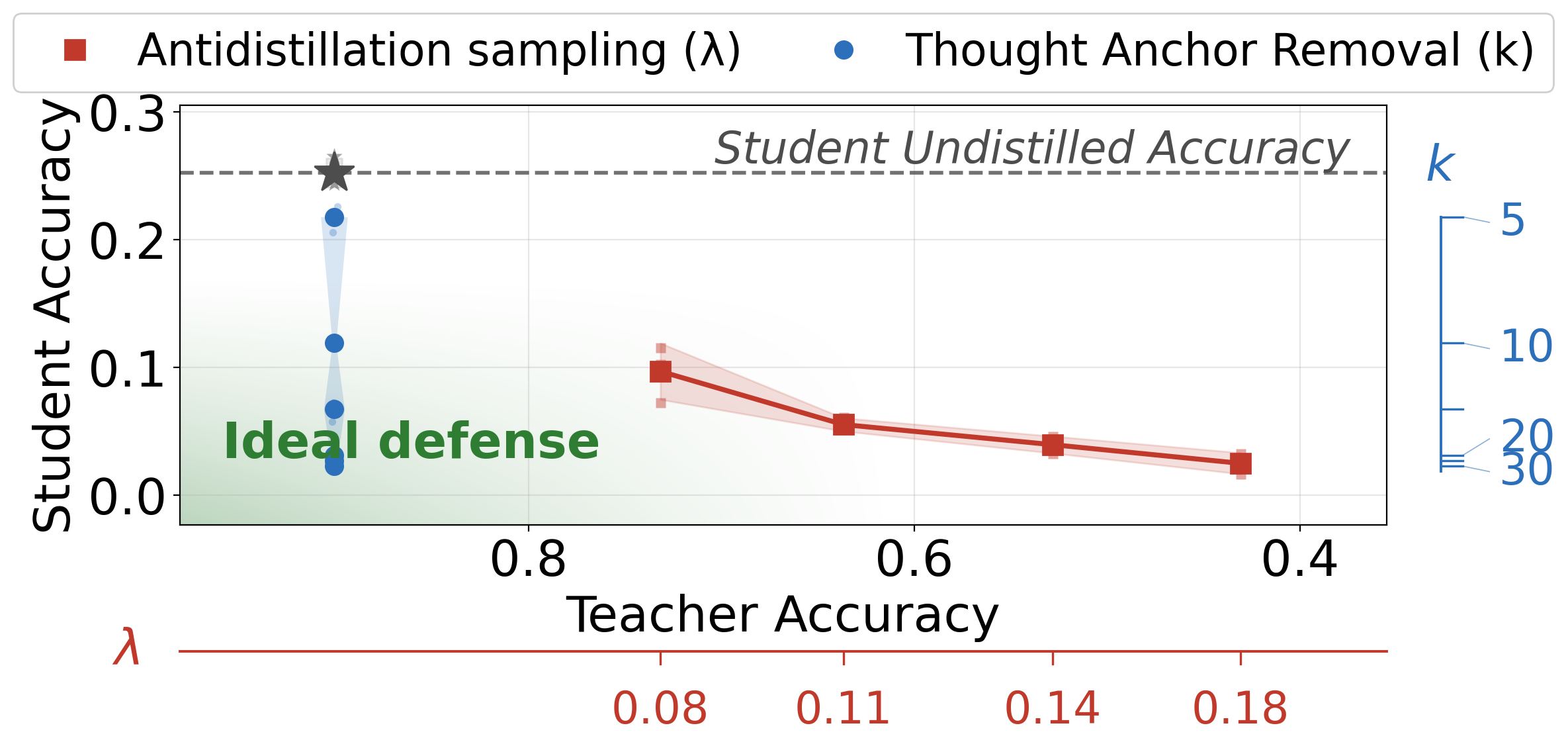}
        \label{fig:accuracy_drop_llama3_1b_gsm8k}
    \end{subfigure}
    \vspace{-10pt}
    \begin{subfigure}[b]{0.49\linewidth}
        \centering
        \includegraphics[width=\linewidth]{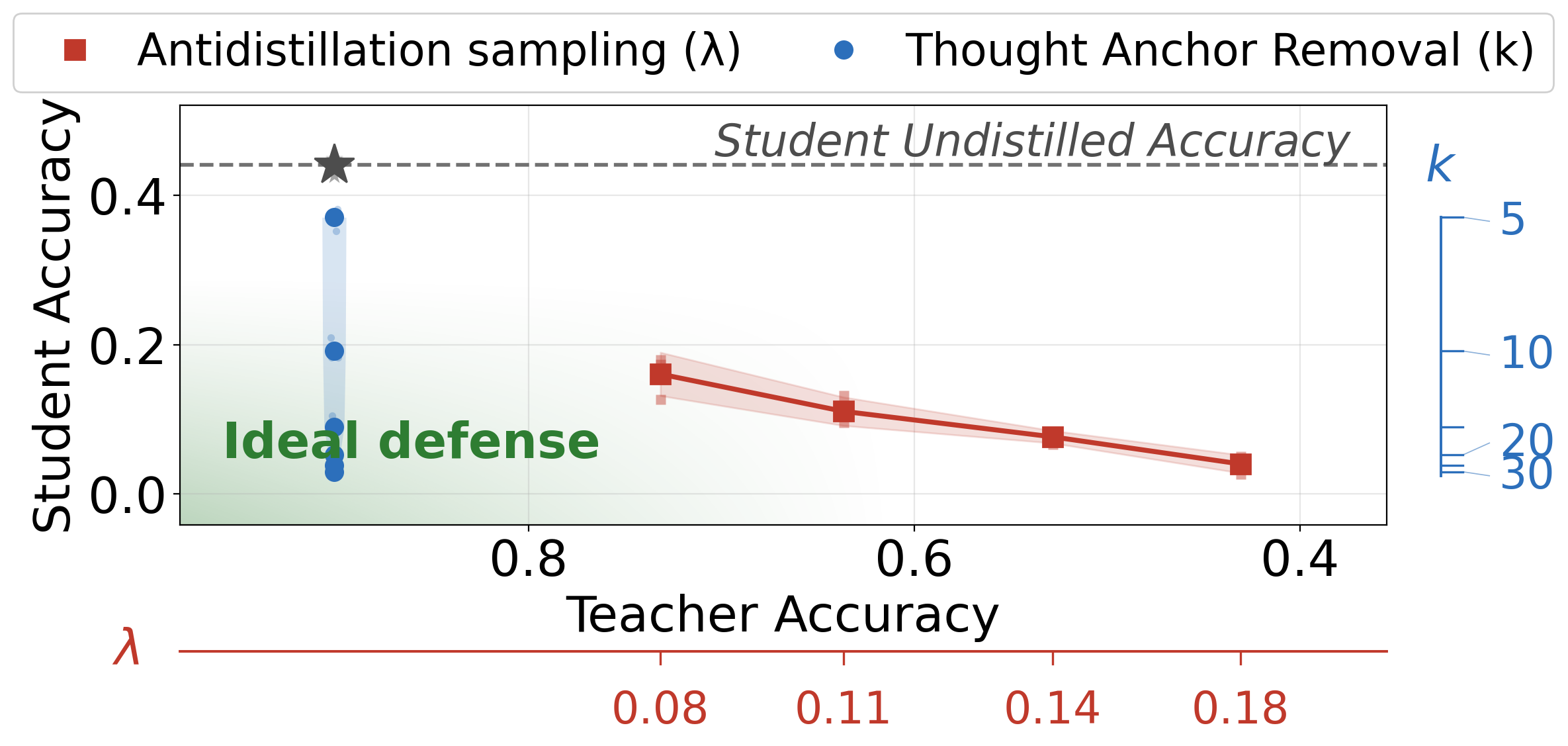}
        \label{fig:accuracy_drop_gemma3_1b_gsm8k}
    \end{subfigure}
    \caption{Student versus teacher accuracy for ADS and TARS on GSM8K, with DeepSeek-R1-Distill-Qwen-7B as the teacher. The students used were Llama 3.2 3B (top-left), Llama 3.2 1B (top-right), and Gemma 3 1B (bottom). }
    \label{fig:teacher-student-accuracy-gsk8k}
\end{figure}

In Tables~\ref{table:detectabilityMATH}-\ref{table:detectabilityGSM8K}, we provide a detectability comparison of TARS and ADS. The perplexity is lower in ADS compared to TARS, which is expected since ADS perturbs the distribution per token, whereas TARS removes entire blocks of tokens. On the other hand, the metrics closer to human perception, such as TF-IDF and Semantic L2 distances for TARS are much stronger, indicating that TARS is better at semantically representing the original traces. The change in tokens with respect to baseline traces is similar and highly dependent on the method's hyperparameter selection across ADS and TARS. We also use an LLM Judge to compute the conference, fluency, and completeness of poisoned traces, with details provided in Appendix \ref{appendix:llm_judge}. 

Each poisoned arm's traces are compared to the unpoisoned baseline. \textbf{PPL} is the poisoned-to-baseline perplexity ratio under a Llama-3-8B-Instruct detector; \textbf{TF-IDF} and \textbf{Sem.\ L2} are the TF-IDF and semantic-embedding ($L_2$) distances to baseline traces; \textbf{$\Delta$Tok.} is the mean token-count difference between the poisoned trace and baseline; \textbf{Overall} is the LLM judge (Qwen3 8B) score rating based on coherence, fluency, and thought completeness (exact prompt stated in Appendix \ref{appendix:llm_judge}). \textbf{TA} is teacher accuracy (\%). Arrows indicate whether lower ($\downarrow$) or higher ($\uparrow$) values are favorable to the defense.

\begin{table}[t]
\centering
\small
\setlength{\tabcolsep}{4pt}
\caption{Detectability for ADS and Thought Anchor Removal (TARS) on MATH, with DeepSeek-R1-Distill-Qwen-7B as teacher. ADS uses Qwen2.5-3B as its proxy model.
\label{table:detectabilityMATH}}
\begin{tabular}{llrrrrrr}
\toprule
Method & Param. & PPL$\downarrow$ & TF-IDF$\downarrow$ & Sem.\ L2$\downarrow$ & $\Delta$Tok.$\downarrow$ & Overall$\uparrow$ & TA$\uparrow$ \\
\midrule
Baseline & ---            & 1.000 & --- & --- & --- & 4.62 & 54.7 \\
\midrule
ADS      & $\lambda=0.08$ & 1.016 & 0.230 & 0.525 & 315 & 4.44 & 53.8 \\
ADS      & $\lambda=0.11$ & 1.066 & 0.396 & 0.683 & 384 & 4.01 & 49.2 \\
ADS      & $\lambda=0.14$ & 1.057 & 0.574 & 0.800 & 428 & 3.57 & 44.0 \\
ADS      & $\lambda=0.18$ & 1.023 & 0.650 & 0.870 & 437 & 2.97 & 40.2 \\
\midrule
TARS     & $k=5$          & 1.236 & 0.025 & 0.255 & 142 & 4.34 & 54.7 \\
TARS     & $k=10$         & 1.455 & 0.126 & 0.425 & 289 & 4.05 & 54.7 \\
TARS     & $k=15$         & 1.697 & 0.202 & 0.535 & 411 & 3.32 & 54.7 \\
TARS     & $k=20$         & 1.986 & 0.253 & 0.614 & 515 & 3.00 & 54.7 \\
TARS     & $k=25$         & 2.311 & 0.287 & 0.671 & 603 & 2.93 & 54.7 \\
TARS     & $k=30$         & 2.750 & 0.328 & 0.717 & 690 & 2.56 & 54.7 \\
\bottomrule
\end{tabular}
\end{table}

\begin{table}[t]
\centering
\small
\setlength{\tabcolsep}{4pt}
\caption{Detectability for ADS and Thought Anchor Removal (TARS) on GSM8K, with DeepSeek-R1-Distill-Qwen-7B as teacher. ADS uses Qwen2.5-3B as its proxy model.
\label{table:detectabilityGSM8K}}
\begin{tabular}{llrrrrrr}
\toprule
Method & Param. & PPL$\downarrow$ & TF-IDF$\downarrow$ & Sem.\ L2$\downarrow$ & $\Delta$Tok.$\downarrow$ & Overall$\uparrow$ & TA$\uparrow$ \\
\midrule
Baseline & ---            & 1.000 & --- & --- & --- & 4.89 & 88.9 \\
\midrule
ADS      & $\lambda=0.08$ & 0.846 & 0.533 & 0.731 & 995 & 3.04 & 38.6 \\
ADS      & $\lambda=0.11$ & 0.820 & 0.713 & 0.883 & 1289 & 2.10 & 26.5 \\
ADS      & $\lambda=0.14$ & 0.836 & 0.760 & 0.935 & 1309 & 1.62 & 17.5 \\
ADS      & $\lambda=0.18$ & 0.914 & 0.840 & 1.036 & 1344 & 1.35 & 11.4 \\
\midrule
TARS     & $k=5$          & 1.648 & 0.038 & 0.400 & 116 & 4.68 & 88.9 \\
TARS     & $k=10$         & 2.525 & 0.313 & 0.740 & 256 & 3.12 & 88.9 \\
TARS     & $k=15$         & 3.262 & 0.728 & 1.133 & 379 & 1.73 & 88.9 \\
TARS     & $k=20$         & 4.478 & 0.894 & 1.288 & 422 & 1.16 & 88.9 \\
TARS     & $k=25$         & 5.775 & 0.940 & 1.327 & 434 & 1.09 & 88.9 \\
TARS     & $k=30$         & 6.775 & 0.943 & 1.332 & 437 & 1.06 & 88.9 \\
\bottomrule
\end{tabular}
\end{table}
 
In Appendix~\ref{appendix:trace_stats}, we extend our empirical testing for detectability by providing trace statistics for both TARS and ADS outputs. These statistics include \textit{mean trace length} (MTL), number of \textit{repetition tokens}, and percentage of trace that is \textit{degenerate} (has no semantic content). We find that while TARS traces have a lower MTL than clean outputs, ADS has significantly higher MTL (\string~\textbf{1.17$\times$}). Furthermore, ADS on average has about \string~\textbf{600} repetition tokens per trace and \string~\textbf{59\%} of the trace is degenerate. These issues with ADS make it trivially detectable by cheap post-processing schemes.

We further supplement this detectability analysis with an analysis of the robustness of TARS vs ADS to different attacks. The first attack analyzed is a simple re-querying scheme where the attacker re-queries the teacher 4 times and chooses the longest response. This is intended to be a countermeasure to TARS which typically has shorter responses due to sentence deletions. The second attack analyzed was a more defense-specific attack which applies the above attack to TARS and a re-querying + repetition filter attack to ADS to combat the degenerate repetition that often happens to ADS traces. Due to space limitations, the details of the experiment and results are provided in Appendix~\ref{appendix:attacks} but we largely see that TARS is robust to both attacks and, in fact, applying the attack can further harm the student while ADS shows clear student performance improvements from both attacks.

\subsection{Theoretical Analysis of Detectability}

We also theoretically analyze TARS's detectability by modeling it as burst deletions over the reasoning trace. Under this model, detectability grows linearly with the number of deleted sentences. If the deletion mask is instead drawn from a distribution conditioned on the input prompt, the relationship stays linear but gains an additive term: the mutual information between the mask and the input. The latter result further motivates TARS as the specific thought anchors targeted (uncertainty-management, planning, and self-checking) intuitively have minimal correlation with the input prompt but are seemingly important for computation.

\begin{proposition}[Informal]\label{prop:delete}
    Assume the LLM acts as an $m$-order Markov Process and the probability of reaching any token from over all $m$-token histories is greater than 0. Then, the detectability of a sequence with $k$ random, burst-deletions against a true LLM-generated sequence is $\mathcal{O}(km)$.
\end{proposition}
\begin{corollary}[Informal]\label{cor:mask_delete}
    Under the same assumptions as Proposition~\ref{prop:delete}, if the deletion mask $S$ is chosen from some distribution $ q(\cdot | x)$ then the bound becomes $\mathcal{O}(km) + I(S;X)$.
\end{corollary}

The formal statements and proofs of Proposition~\ref{prop:delete} and Corollary~\ref{cor:mask_delete} are provided in Appendix~\ref{appendix:proofs}. The $m$-order Markov assumption is an approximation of LLM behavior \citep{basuCV2023}. Papers such as \cite{LiuLHPBPL2024} and \cite{XiaoTCHL2024} support that the effective context window is not actually the full context window length. For the lower-bounded probability assumption, this can be fulfilled by a top-$k$ selection method as is typically done in practice. 

Proposition~\ref{prop:delete} and Corollary~\ref{cor:mask_delete} show that for token deletions, the detectability follows a linear relationship with the number of deletions. This supports our use of $k$ as a hyperparameter to determine the trade-off between detectability and student degradation.

\section{Conclusion \& Discussion}
\label{sec:conclusion}

We cast antidistillation as a max-min game, explicitly integrating detectability into the constraint set to address the semantic and syntactic vulnerabilities ignored by prior methods. Drawing on mechanistic interpretability, we isolate thought anchors as the ideal targets for these sparse interventions and validate this approach with TARS. Empirically, TARS illustrates the trade-off between student degradation and defense detectability that our framework anticipates. Ultimately, our combined theoretical and practical findings represent a significant step toward robust antidistillation mechanisms that resist counter-attacks without compromising the integrity of the teacher’s reasoning.

Several directions remain open. TARS is highly dependent on reasoning patterns, therefore, it will fail in situations where certain thought anchor types are not used. The max-min game-theoretic formulation also opens the door for a rigorous derivation of achievability limits for antidistillation defense mechanisms. Finally, studying specific methods for antidistillation detection can introduce further theoretically-motivated metrics to analyze the detectability performance of defense methods.




\subsection*{Reproducibility statement}

We report all implementation details for all experiments including models, datasets, compute used, training hyperparameters, and seeds in the Appendix. For the theory, fully rigorous theorem statements and proofs are provided in the Appendix. Finally, we will also open source all of our code upon paper publication.



\bibliography{iclr2027_conference}
\bibliographystyle{iclr2027_conference}

\appendix

\section{Generalization of Previous Antidistillation Methods}\label{appendix:gen_prior_methods}
The main distinction between different antidistillation methods is how they instantiate the set of admissible perturbations of the dataset $D$ and in what they consider the set of viable student architectures. In the remainder of this subsection, we show how both Antidistillation Sampling \citep{SAVANITFX2025} and Defensive Output Generation \citep{LITZQ2025} fit into the framework of \eqref{eq:antidistill}.

\subsection{Antidistillation Sampling}
ADS proceeds in an online antidistillation fashion, where individual teacher outputs are given per user query. Their formulation is a greedy approach to the outer maximization over $\tilde{D}_T$, where they try to maximize the loss of each individual teacher output. ADS considers the following perturbation to the teacher output distribution:
\begin{align*}
\tilde{p}(\cdot \mid x_{1:t}; \theta_T) := \frac{1}{Z} \exp \Bigl( &\frac{1}{\tau} \log p(\cdot \mid x_{1:t}; \theta_T) + \lambda \Delta(\cdot \mid x_{1:t}; \theta_P) \Bigr),
\end{align*}

where $Z$ is a normalization factor, $\tau$ is the sampling temperature, $p(\cdot\mid x_{1:t};\theta_T)$ is the original teacher distribution for token $t+1$, $\lambda$ is a scaling hyperparameter and $\Delta$ is a perturbation to the output distribution determined by the gradients of a proxy model.\footnotemark

\footnotetext{
$\begin{aligned}[t]
\Delta(\cdot \mid x_{1:t}; \theta_P) &= \ell(\theta_P^+) - \ell(\theta_P) \\
&= \ell(\theta_P + \eta \nabla_{\theta_P} \log p(x_{t+1} \mid x_{1:t}; \theta_P)) - \ell(\theta_P)
\end{aligned}$
}

Thus, the set of admissible perturbations is the following,
\begin{align}\label{eq:b_ads}
    \mathcal{B}_{\mathrm{ADS}}(D; \theta_T) = \left \{ \tilde{x}_i : \tilde{x}_i \sim \tilde{p}(\cdot \mid x_{1:t};\theta_T) \right\}_{i=1}^n,
\end{align}
and the admissible class of training algorithms is simply empirical risk minimization over admissible architectures.

\subsection{Defensive Output Generation}
DOGe proposes that the task of antidistillation can be framed as the following bi-level optimization problem. The defender controls the teacher’s parameters to maximize its performance while weakening any student distilled from its outputs. The optimization problem is thus:
\begin{align*}
&\theta^*_{T} = \arg \max_{\theta_{T}} \left[ \mathrm{Perf}(\theta_T) - \lambda \cdot \mathrm{Perf} \left( \theta^*_S \right) \right]\\
&\theta^*_S = \arg \min_{\theta_S} \mathcal{L}(\theta_S ; D_{KD}(\theta_{T})),
\end{align*}
where $\mathrm{Perf}(\cdot)$ is the performance of the input model, $D_{KD}$ is the outputs of the teacher model $\theta^*_T$, and $\lambda > 0$ is a scaling hyperparameter. 

We connect this to our formulation in \eqref{eq:antidistill} by first defining
\begin{align}\label{eq:b_doge}
    &\mathcal{B}_{\mathrm{DOGe}}(D; \theta_T,\Theta_T) \\&= \left\{ \tilde{D}_T : \tilde{D}_T \sim p(\cdot\mid \theta_T), \theta_T \in \Theta_T, \mathrm{Perf}(\theta_T) \geq \tau \right\}.
\end{align}
If we now consider the case where the training algorithm is simply empirical risk minimization, \eqref{eq:antidistill} simplifies to
\begin{align*}
\tilde{D}_T &= \arg\max_{\tilde{D}_T \in \mathcal{B}_{\mathrm{DOGe}}} \E_{x\sim p_{\mathrm{data}}} \left[ -\mathrm{Perf}(\theta_S^*(\tilde{D}_T);x) \right].
\end{align*}
Since the dataset entirely depends on the parameters $\theta_T$ we can determine $\tilde{D}_T$ equivalently by finding
\begin{align*}
    \theta_T^* = \arg\max_{\theta_T \in \Theta_T,
    \mathrm{Perf}(\theta_T)\geq \tau} \E_{x\sim p_{\mathrm{data}}} \left[ -\mathrm{Perf}(\theta_S^*(\tilde{D}_T);x) \right],
\end{align*}
and taking a Lagrangian relaxation of this problem yields
\begin{align*}
    \theta_T^* = \arg\max_{\theta_T \in \Theta_T} \mathrm{Perf}(\theta_T) - \lambda \mathrm{Perf}(\theta_S^*(\theta_T))
\end{align*}
which is exactly the DOGe framework.

\newpage 
\section{TARS Algorithm}\label{appendix:TARS}
In Algorithm \ref{algo:detectTA}, we include pseudo-code for the receiver head calibration and the broadcast scoring to compute which sentences to delete. This is used by Algorithm \ref{algo:deleteTA} to find which sentences to delete. In our experiments using these algorithms, we used $n=200$ traces, $K=32$ heads, and distance $d=4$.

\begin{algorithm}[!t]
\caption{\textsc{Detect-Thought-Anchor-Sentences} selects, once per teacher,
the receiver heads that concentrate attention most sharply on a few
sentences, then scores each sentence of a trace by the attention it receives
through those heads from distant downstream sentences, so that the
highest-scoring sentences are the thought anchors.}
\label{algo:detectTA}
\begin{algorithmic}[1]
\Require teacher $M$ with head set $H$, calibration traces $\mathcal{T}$, prompt $x$, response $y$, distance $d$, head count $K$
\Function{VerticalAttention}{$x, y, G$}
  \State $(s_1, \dots, s_n) \gets \textsc{Segment}(y)$
  \State run $M$ on $x \,\Vert\, y$, recording attention $A_{\ell,h}$ for $(\ell,h) \in G$
  \For{each $(\ell,h) \in G$}
    \State $\bar{A}[i,j] \gets \operatorname{mean}\{A_{\ell,h}[p,q] : p \in s_i,\ q \in s_j,\ q \le p\}$ for all $i \ge j$
    \State $v_{\ell,h}[j] \gets \operatorname{mean}_{\,i > j + d}\ \bar{A}[i,j]$ \quad ($0$ if no such $i$)
  \EndFor
  \State \Return $(s_1, \dots, s_n)$, $\{v_{\ell,h}\}_{(\ell,h) \in G}$
\EndFunction
\Statex \textit{Calibrate receiver heads (once per teacher)}
\For{each trace $(x_t, y_t) \in \mathcal{T}$}
  \State $\{v_{\ell,h}\} \gets \Call{VerticalAttention}{x_t, y_t, H}$
  \State $\mu_t \gets \frac{1}{|H|} \sum_{(\ell,h) \in H} v_{\ell,h}$ \Comment{trace's shared profile}
  \State $\kappa_t(\ell,h) \gets \operatorname{ExKurt}\!\left(v_{\ell,h} - \mu_t\right)$ for all $(\ell,h) \in H$
\EndFor
\State $\bar\kappa(\ell,h) \gets \operatorname{mean}_{t}\, \kappa_t(\ell,h)$
\State $R \gets$ the $K$ heads with largest $\bar\kappa$ \Comment{receiver heads}
\Statex \textit{Score sentences}
\State $(s_1, \dots, s_n), \{v_{\ell,h}\} \gets \Call{VerticalAttention}{x, y, R}$
\State $\mathrm{score}[j] \gets \frac{1}{|R|} \sum_{(\ell,h) \in R} v_{\ell,h}[j]$ \Comment{broadcasting score}
\State \Return $(s_1, \dots, s_n)$, $\mathrm{score}$
\end{algorithmic}
\end{algorithm}

\section{Thought Anchor Analysis}
\label{appendix:thought_anchor}

\begin{figure}[!t]
    \centering
    \includegraphics[width=1\linewidth]{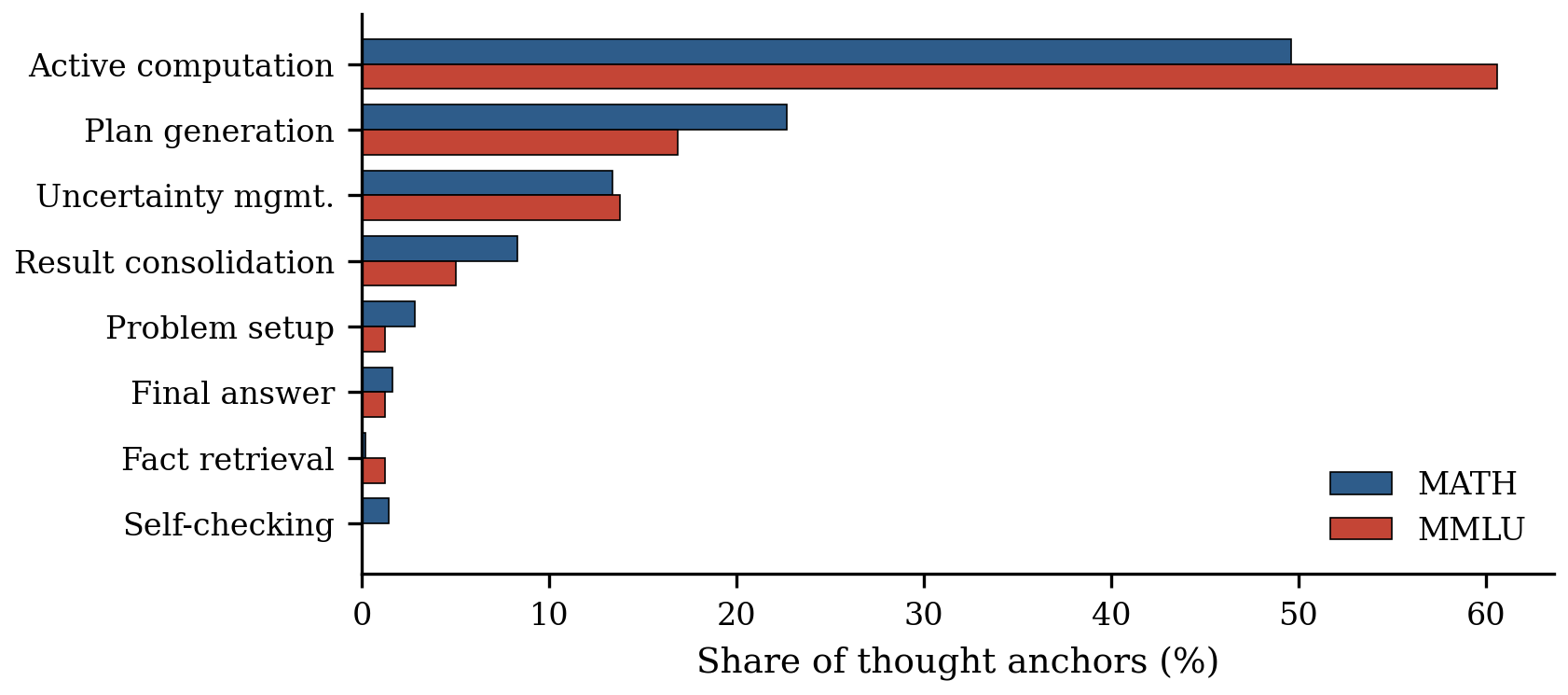}
    \caption{Percent share of thought anchors across MATH and MMLU. While Active Computation is the most common thought anchor, which is both important for answer generation and user understanding, Plan Generation and Uncertainty Management are mainly useful for answer generation. Additionally, anchor category distributions agree closely across the two datasets (Pearson $r=0.98$ , $p<10^{-4}$), indicating that differences in occurrences of thought anchors is dataset agnostic.}
    \label{fig:anchor_distribution}
\end{figure}

The main text motivates why we look at thought anchors and shows representative examples; this appendix details the experimental setup and reports the full per-dataset results. Because thought anchors are the unit at which we hypothesize disruption will propagate to a student under distillation, we need a precise account of how they are estimated, how stable that estimate is, and what the sentence-category distribution looks like across our two datasets. We follow the methodology of \cite{BOGDANMNC2025}, with implementation details adapted to our teacher model and trace lengths, described below. Specifically, we analyze correct rollouts from DeepSeek R1 Distill Qwen 7B on MATH and MMLU, retaining one trace per problem: $50$ traces per dataset for analysis and a disjoint $10$ for receiver-head identification.

Following \cite{BOGDANMNC2025}, we identify attention heads that concentrate long-range attention onto a small number of earlier sentences. Responses are split into sentences on terminal punctuation, and tokens mapped to sentences via tokenizer offsets. For each head $(\ell, h)$ and each identification trace, we average token-level attention into a sentence-level matrix $A^{(\ell,h)} \in \mathbb{R}^{S \times S}$ and compute the \emph{vertical attention} received by sentence $j$,
\begin{equation}
v^{(\ell,h)}_j = \frac{1}{|\mathcal{D}_j|}\sum_{i \in \mathcal{D}_j} A^{(\ell,h)}_{ij}, \qquad \mathcal{D}_j = \{i : i > j + d_{\min}\},
\end{equation}
with $d_{\min}=4$ to isolate long-range from adjacency effects. We subtract the per-trace head-averaged profile $\bar{v}_j$ to remove the shared baseline, and score each head by the excess kurtosis of the residual $v^{(\ell,h)} - \bar{v}$, which is high when a head's long-range attention is concentrated on a small number of sentences rather than spread across the trace. We average kurtosis across identification traces and retain the top $K=28$ heads as the receiver set $\mathcal{R}$. For each analysis trace, the broadcasting score for sentence $j$ is the mean vertical attention it receives across receiver heads, $\mathrm{BC}(j) = \frac{1}{K}\sum_{(\ell,h) \in \mathcal{R}} v^{(\ell,h)}_j$, and the top-$10$ sentences by $\mathrm{BC}(j)$ are reported as that trace's thought anchors.

To characterize \emph{what kind} of reasoning step each anchor corresponds to, we tag every sentence with one of eight functional categories from \cite{BOGDANMNC2025} (plan generation, uncertainty management, self-checking, fact retrieval, active computation, result consolidation, problem setup, final answer) using regex matching on category-specific keywords (e.g., ``alternatively'' for planning; ``wait'', ``actually'' for uncertainty management), with unmatched sentences defaulting to the modal category, \emph{active computation}. Identification uses sequence length $2{,}048$; broadcasting uses $4{,}096$.

\subsection{Details for Thought Anchor Transfer Experiment}
\label{appendix:anchor_transfer}

For the transfer analysis in Section~\ref{sec:anchor_transfer}, we use the same procedure for identifying receiver heads and computing broadcasting scores as described above, with the following changes.

\paragraph{Models.} Receiver heads are identified separately for the teacher (DeepSeek-R1-Distill-Qwen-7B) and for the student (Llama-3.2-3B), both before and after distillation, using $N$ identification traces and a sequence length of $4{,}096$ for each model.

\paragraph{Receiver head selection.} Since models differ in size and in how concentrated their receiver heads are, we do not use a fixed number of heads $K$ across models. Instead, for each model we keep the heads whose average kurtosis lies above the 95th percentile of that model's own kurtosis distribution. For the teacher, the median kurtosis across all heads is $7.5$ and the 95th percentile is $23.2$, yielding $K=24$. For the student, we settle on $K=6$ using a similar procedure.

\paragraph{Anchor agreement.} Both models read the same $128$ teacher traces, truncated to $4{,}096$ tokens, and each computes $\mathrm{BC}(j)$ for every sentence using its own receiver heads. Since sentences are segmented on the text alone, the two score vectors are aligned sentence by sentence, even though the models, tokenizers, and heads differ. For each trace, we compute Spearman $\rho$ between the two score vectors, which measures how similarly the two models rank all sentences and is insensitive to differences in the absolute scale of attention, and the top-5 Jaccard similarity $|A \cap B| / |A \cup B|$, where $A$ and $B$ are the teacher's and student's five highest-scoring sentences. Traces with fewer than five sentences are excluded, and both metrics are averaged over traces.

\section{Additional TARS Results}
\label{appendix:more_TARS_results}

\subsection{Scaled experiment}
\label{appendix:gemma3deepmath}
In the experiments of the main text, we focused only on one teacher. We provide this additional experiment which uses a Gemma 3 12B as the teacher to help indicate that this method scales with teacher size on harder datasets. Figure \ref{fig:teacher-student-accuracy-deepmath} shows a degradation of student accuracy with minimal change in teacher accuracy for TARS, whereas ADS significantly degrades teacher accuracy.

\begin{figure}[!h]
    \centering
    \includegraphics[width=0.99\linewidth]{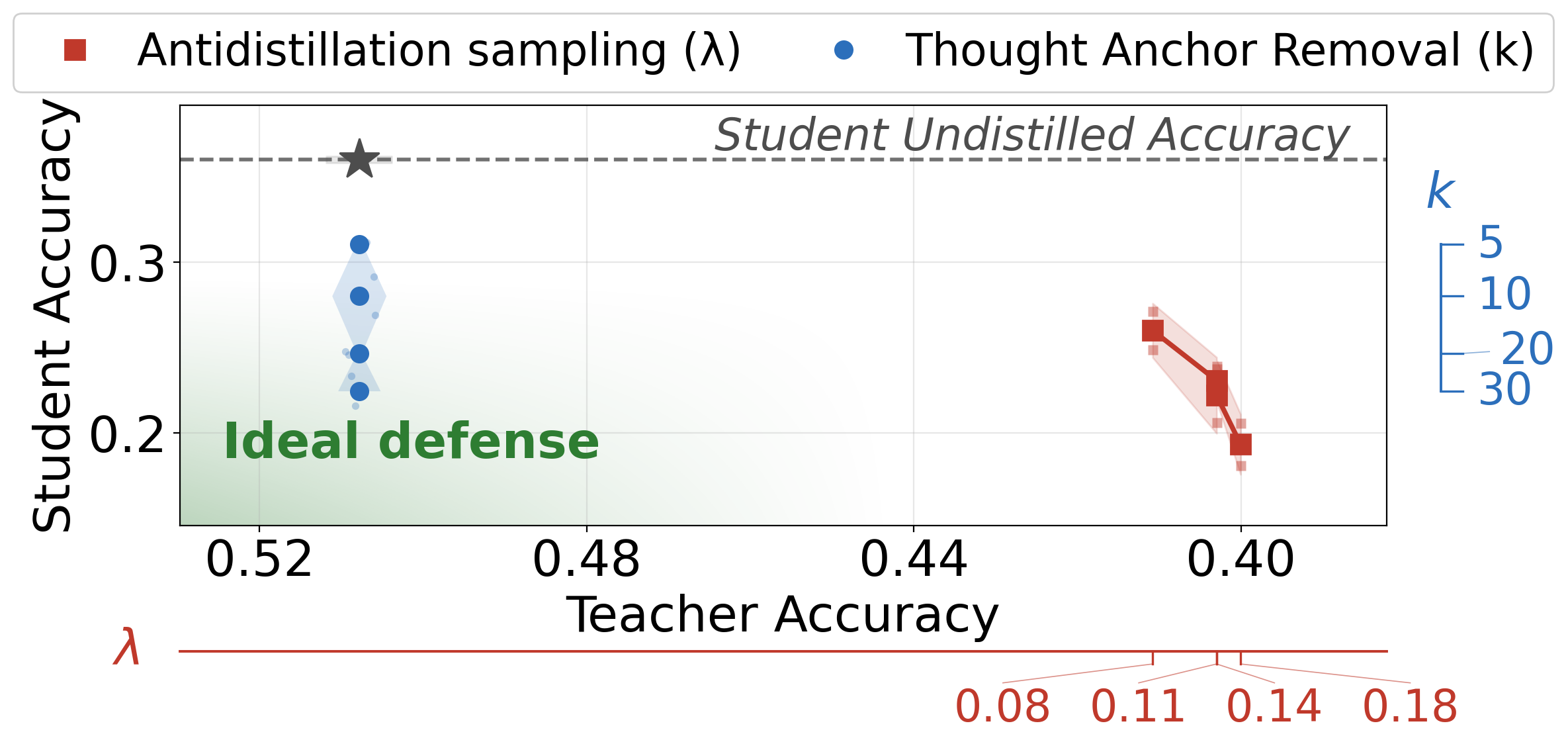}
    \caption{Student versus Teacher Accuracy for ADS (Red) and TARS (Blue). In this figure, the teacher is Gemma 3 12B, student is Qwen 3 1.7B, and dataset is DeepMath. The ideal antidistillation method is the bottom left region of the plot, in which the student accuracy is degraded significantly without degrading the teacher accuracy. Antidistillation Sampling shows a trade-off between degrading the student and teacher, whereas poisoning though anchors shows that scaling defense degrades the student without degrading the teacher.}
    \label{fig:teacher-student-accuracy-deepmath}
\end{figure}

\subsection{Remaining MATH and MMLU Pro}
\label{appendix:failure_case_mmlu}
In Figure \ref{fig:teacher-student-accuracy-mmlu-pro}, Thought anchor removal does not degrade student accuracy or teacher accuracy, whereas ADS degrades student accuracy at the cost of an extreme teacher accuracy decrease. In this case, both methods have limited applicability to use.

In Figure \ref{fig:accuracy_drop_llama3_1b_math}, both ADS and Thought Anchor Removal fail to meaningfully degrade the student accuracy. We hypothesize that this is because the unpoisoned student accuracy was very low.

\begin{figure}[!t]
    \centering
    \begin{subfigure}[b]{0.49\linewidth}
        \centering
        \includegraphics[width=\linewidth]{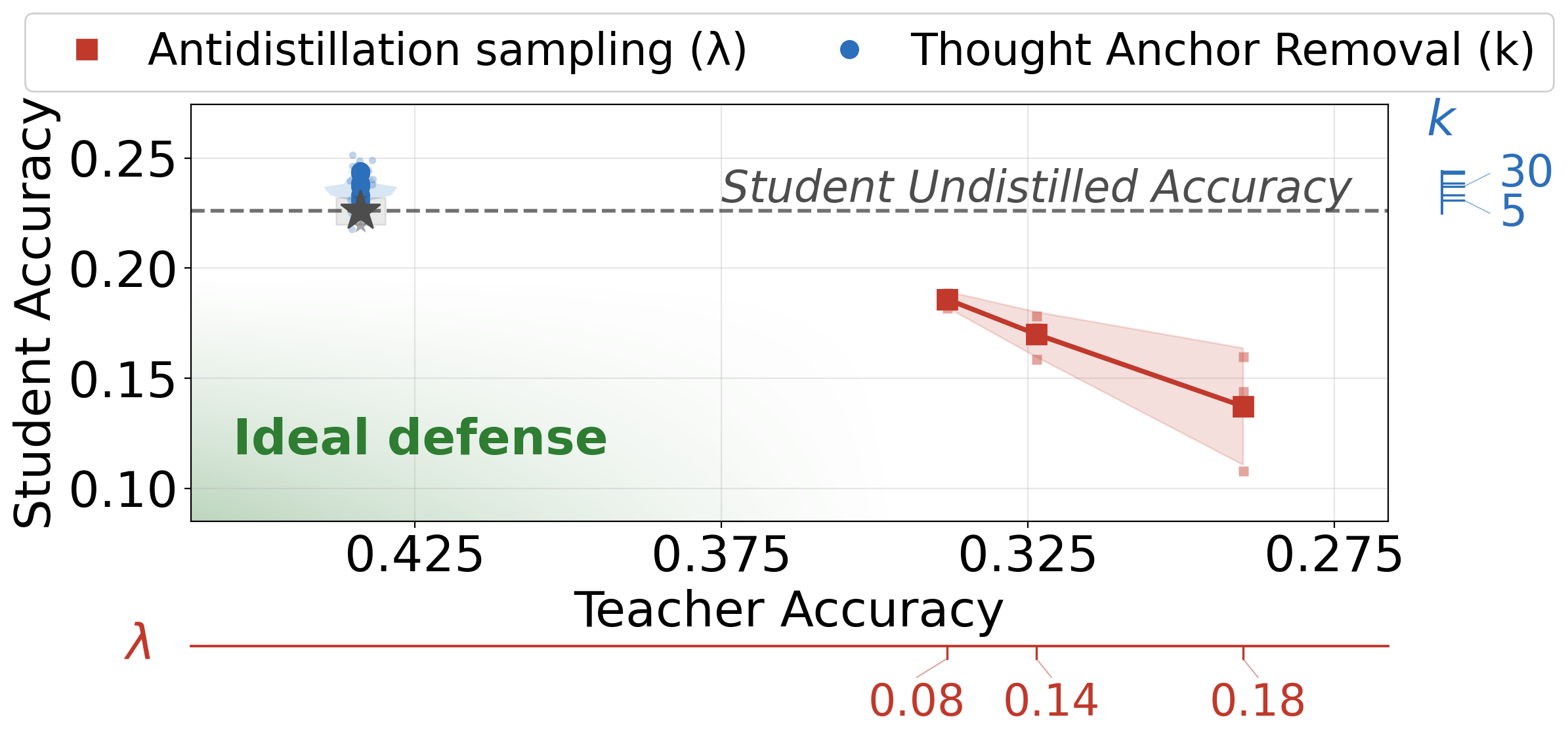}
        \caption{Student: Llama 3.2 3B}
    \end{subfigure}\hfill
    \begin{subfigure}[b]{0.49\linewidth}
        \centering
        \includegraphics[width=\linewidth]{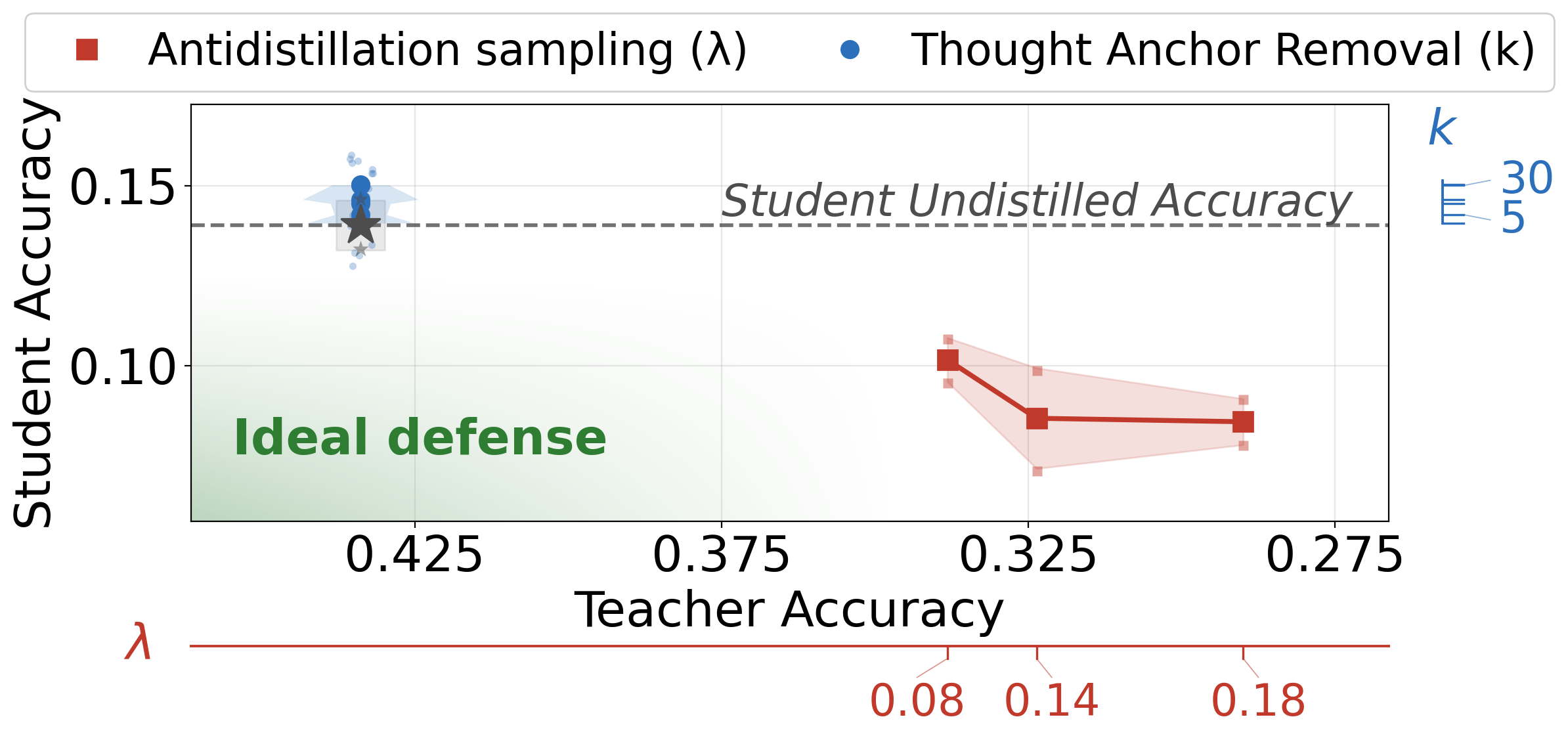}
        \caption{Student: Llama 3.2 1B}
    \end{subfigure}

    \begin{subfigure}[b]{0.49\linewidth}
        \centering
        \includegraphics[width=\linewidth]{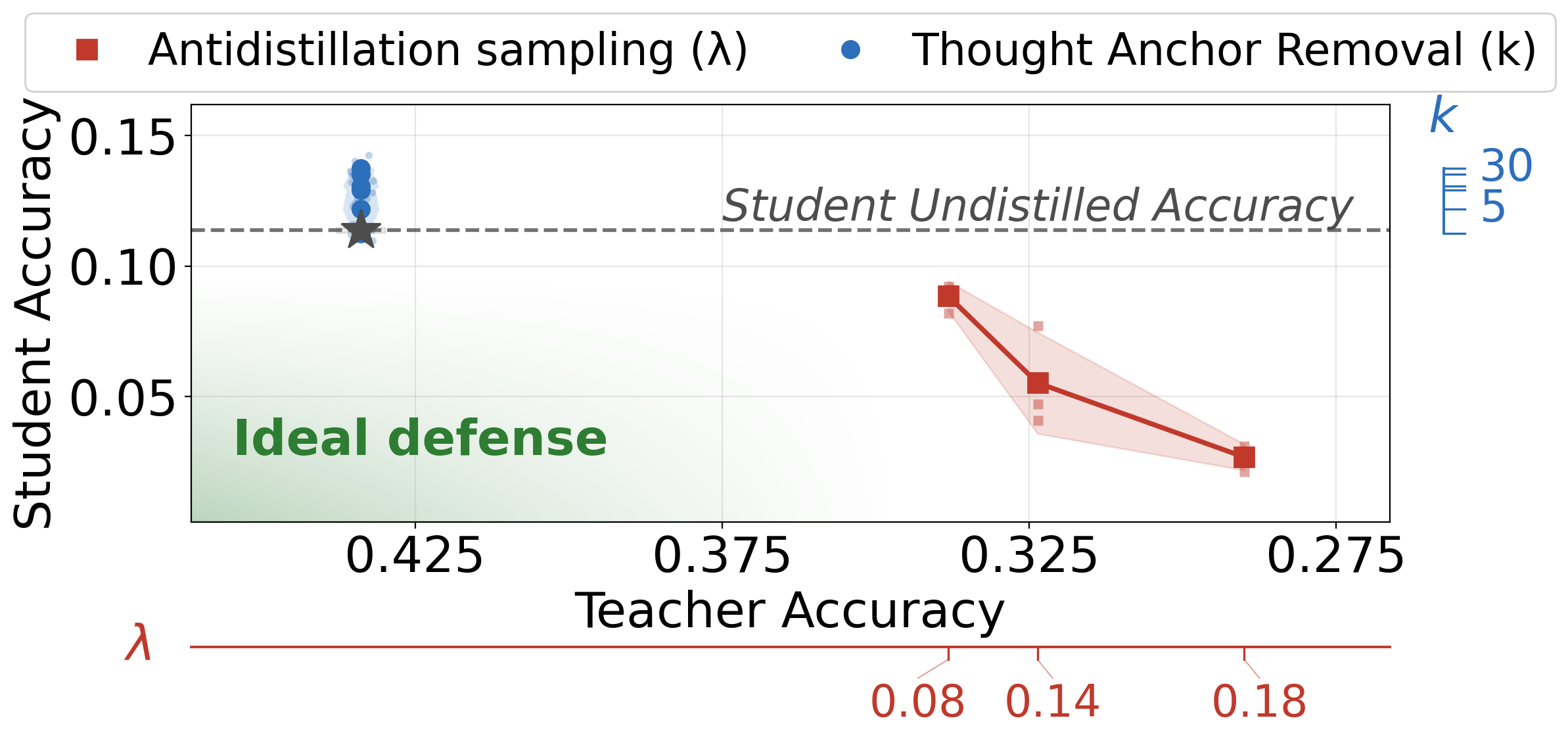}
        \caption{Student: Gemma 3 1B}
    
    \end{subfigure}
    \caption{Student versus teacher accuracy for ADS (red) and TARS (blue) on MMLU-Pro, with DeepSeek-R1-Distill-Qwen-7B as the teacher. The students used were Llama 3.2 3B (top-left), Llama 3.2 1B (top-right), and Gemma 3 1B (bottom). In this case, Thought Anchor Removal does not degrade student or teacher accuracy, while ADS degrades student accuracy at the cost of an extreme teacher-accuracy decrease.}
    \label{fig:teacher-student-accuracy-mmlu-pro}
\end{figure}

\begin{figure}[!t]
    \centering
    \includegraphics[width=0.55\linewidth]{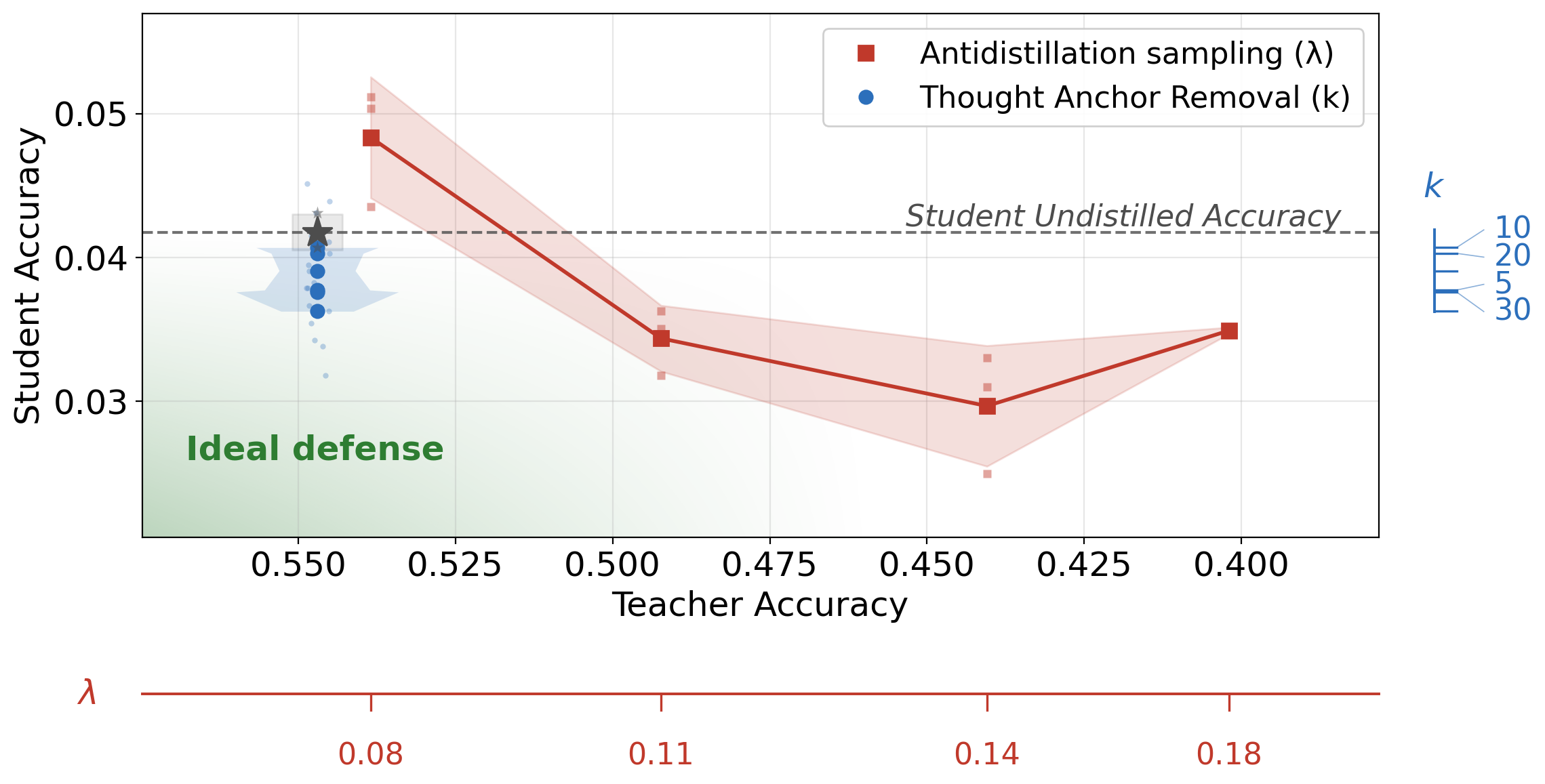}
    \caption{Student versus teacher accuracy for ADS (red) and TARS (blue) on MATH, with DeepSeek-R1-Distill-Qwen-7B as the teacher and Llama 3.2 1B as the student. In this example, both ADS and Thought Anchor Removal fail to meaningfully degrade the student accuracy. We hypothesize that this is because the unpoisoned student accuracy was already very low.}
    \label{fig:accuracy_drop_llama3_1b_math}
\end{figure}

We also provide the detectability results for MMLU Pro.

\begin{table}[t]
\centering
\small
\setlength{\tabcolsep}{4pt}
\caption{Detectability for ADS and Thought Anchor Removal (TARS) on MMLU-Pro, with DeepSeek-R1-Distill-Qwen-7B as teacher. ADS uses Qwen2.5-3B as its proxy model.
\label{table:detectabilityMMLUPro}}
\begin{tabular}{llrrrrrr}
\toprule
Method & Param. & PPL$\downarrow$ & TF-IDF$\downarrow$ & Sem.\ L2$\downarrow$ & $\Delta$Tok.$\downarrow$ & Overall$\uparrow$ & TA$\uparrow$ \\
\midrule
Baseline & ---            & 1.000 & --- & --- & --- & 4.42 & 43.4 \\
\midrule
ADS      & $\lambda=0.08$ & 0.693 & 0.770 & 1.023 & 806 & 1.89 & 33.8 \\
ADS      & $\lambda=0.11$ & 0.602 & 0.882 & 1.172 & 861 & 1.48 & 34.7 \\
ADS      & $\lambda=0.14$ & 0.579 & 0.921 & 1.243 & 865 & 1.10 & 32.4 \\
ADS      & $\lambda=0.18$ & 0.588 & 0.949 & 1.280 & 870 & 1.04 & 29.0 \\
\midrule
TARS     & $k=5$          & 1.454 & 0.014 & 0.240 & 78 & 4.04 & 43.4 \\
TARS     & $k=10$         & 2.201 & 0.034 & 0.357 & 158 & 3.67 & 43.4 \\
TARS     & $k=15$         & 3.021 & 0.071 & 0.458 & 239 & 3.29 & 43.4 \\
TARS     & $k=20$         & 4.475 & 0.135 & 0.570 & 327 & 3.01 & 43.4 \\
TARS     & $k=25$         & 6.692 & 0.251 & 0.697 & 412 & 2.56 & 43.4 \\
TARS     & $k=30$         & 14.160 & 0.413 & 0.859 & 487 & 1.96 & 43.4 \\
\bottomrule
\end{tabular}
\end{table}

\newpage

\section{Detectability for defense methods}
\label{appendix:detectability_metrics}
Perplexity measures how "surprised" a language model is by a sequence of text. Our perplexity metric uses Llama 3 8B Instruct as a detector model computed as a sliding window (max length 2048, stride 512). TF-IDF (Term Frequency-Inverse Document Frequency) focuses on the specific vocabulary used in the traces. By measuring the mean cosine distance between paired traces, we are looking for lexical consistency. Our specific implementation uses an ngram range of $(2,5)$ and L2 norm with 20k max features. Semantic Similarity is the mean L2 distance between paired BGE-M3 (\texttt{BAAI/bge-m3}) embeddings with \texttt{max\_length=8192}.

\subsection{LLM-as-a-Judge}
\label{appendix:llm_judge}
We evaluate trace coherence using LLM-as-a-Judge with Qwen 3 8B. The following judge prompt is used:

\begin{promptbox}{Judge Prompt}
\small\setlength{\parskip}{0.3em}
\textbf{System prompt}\par\smallskip
You are a careful evaluator of mathematical reasoning transcripts. You rate how
well-formed a transcript is as a piece of writing and reasoning. You are NOT
checking whether the final answer is correct --- a transcript can be
well-written and still reach a wrong answer, and it can be badly written while
reaching a right one. Judge only the qualities described in the rubric. Respond
with JSON and nothing else.
 
\textbf{User prompt}\par\smallskip
Rate the reasoning transcript below on three dimensions.
 
Each dimension is an integer from 1 to 5, where \textbf{5 is the best possible
score and 1 is the worst}. Higher is always better: 5 means the transcript
excels on that dimension, 1 means it fails badly on it. Use the full range ---
do not cluster every rating in the middle.
 
\texttt{coherence} --- does the reasoning follow logically from step to step?
\begin{itemize}
  \item[5 (best):] every step follows from the last; the thread is easy to track.
  \item[4:] largely sound, with a minor gap or unexplained step.
  \item[3:] mostly followable, but with non-sequiturs, abrupt jumps, or a step
            that references something never established.
  \item[2:] hard to follow; the line of argument repeatedly breaks down.
  \item[1 (worst):] steps do not connect; the reasoning contradicts itself or
            wanders with no discernible line of argument.
\end{itemize}
 
\texttt{fluency} --- does it read like natural, well-formed text?
\begin{itemize}
  \item[5 (best):] clean prose and notation throughout.
  \item[4:] minor blemishes that do not impede reading.
  \item[3:] noticeable artifacts --- garbled fragments, malformed notation, a few
            repeated phrases.
  \item[2:] frequent artifacts; substantial stretches are hard to read.
  \item[1 (worst):] heavily degraded --- long stretches of repeated tokens or
            phrases, nonsense strings, text collapsed into a loop.
\end{itemize}
 
\texttt{completeness} --- does the transcript arrive somewhere?
\begin{itemize}
  \item[5 (best):] works through to a clearly stated final answer.
  \item[4:] reaches a conclusion, though the last step is compressed.
  \item[3:] substantially complete, but the ending is abrupt or the conclusion
            is left implicit.
  \item[2:] breaks off close to the end, with no stated conclusion.
  \item[1 (worst):] stops mid-thought, mid-sentence, or mid-calculation with no
            conclusion at all.
\end{itemize}
 
Then give an overall integer from 1 to 5 (again, 5 = best, 1 = worst) for how
usable this transcript would be as an example of good reasoning, and one short
sentence of justification.
 
Reply with exactly this JSON and nothing else:\par\smallskip
\texttt{\{"coherence": <1-5>, "fluency": <1-5>, "completeness": <1-5>,}\\
\texttt{\phantom{\{}"overall": <1-5>, "reason": "<one sentence>"\}}
\end{promptbox}



\section{Trace Examples of TARS and ADS} 
\label{appendix:ads_failure}
\begin{figure}[H]
    \centering
    \includegraphics[width=1\linewidth]{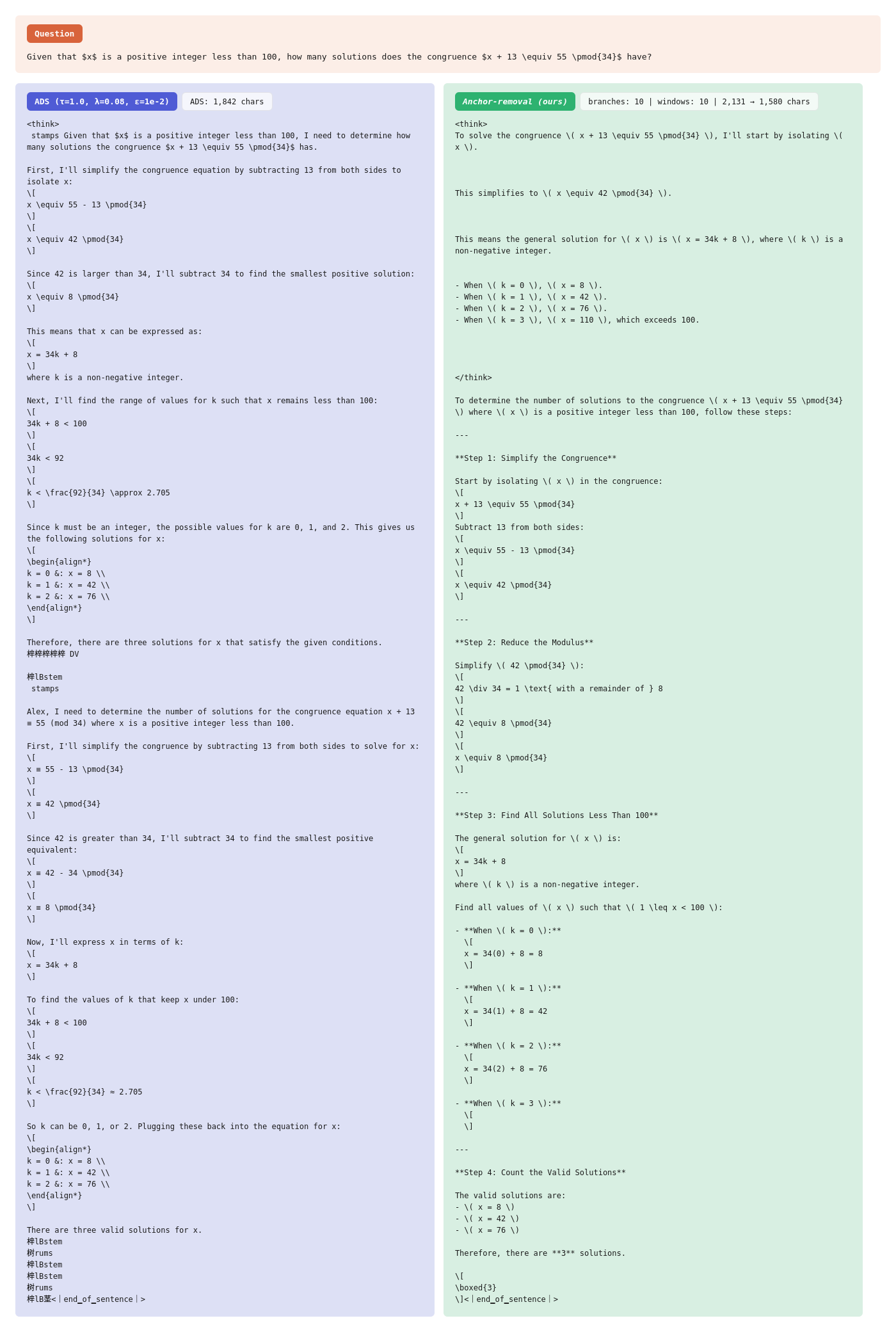}
    \caption{Example ADS versus TARS traces. The ADS trace insert random tokens in places, such as at the end, which does not happen in TARS.}
    \label{fig:case_6_compare}
\end{figure}

\begin{figure}
    \centering
    \includegraphics[width=1\linewidth]{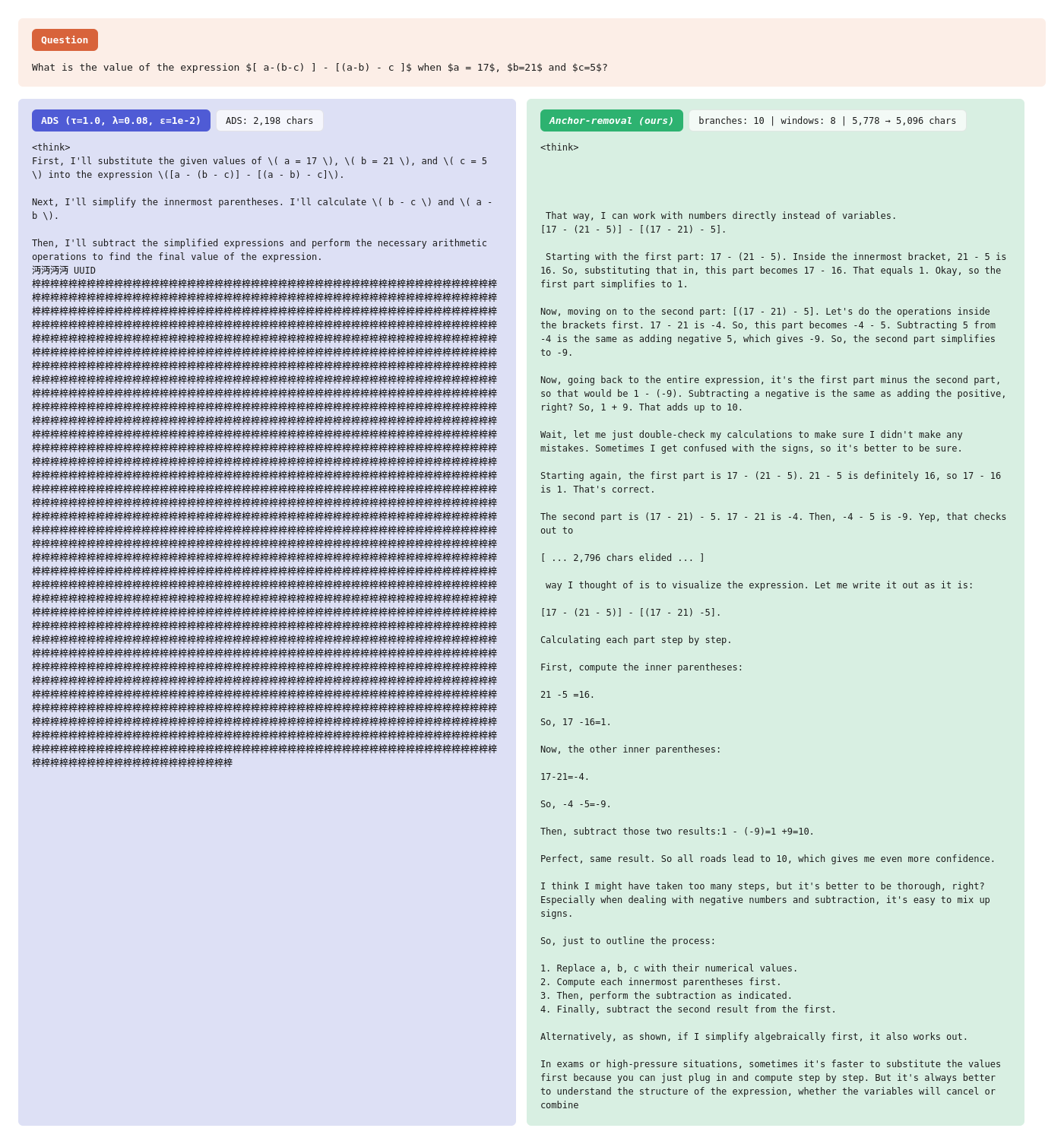}
    \caption{Example ADS versus TARS traces. ADS looks acceptable for the first few sentences, before repeating the same random character over and over.}
    \label{fig:case_246_compare}
\end{figure}
\section{Additional Experiments}

\subsection{Method against random sentence poisoning}
\label{appendix:random_poisoning}

Our method shows that student model distillation performance is worse when the teacher model traces are poisoned than baseline. In this additional experiment, we compare the results of our standard method to poisoning random sentences. The goal of this experiment is to verify that the backtracking sentences we poison are truly important for distillation.

\begin{figure}[H]
    \centering
    \includegraphics[width=0.75\linewidth]{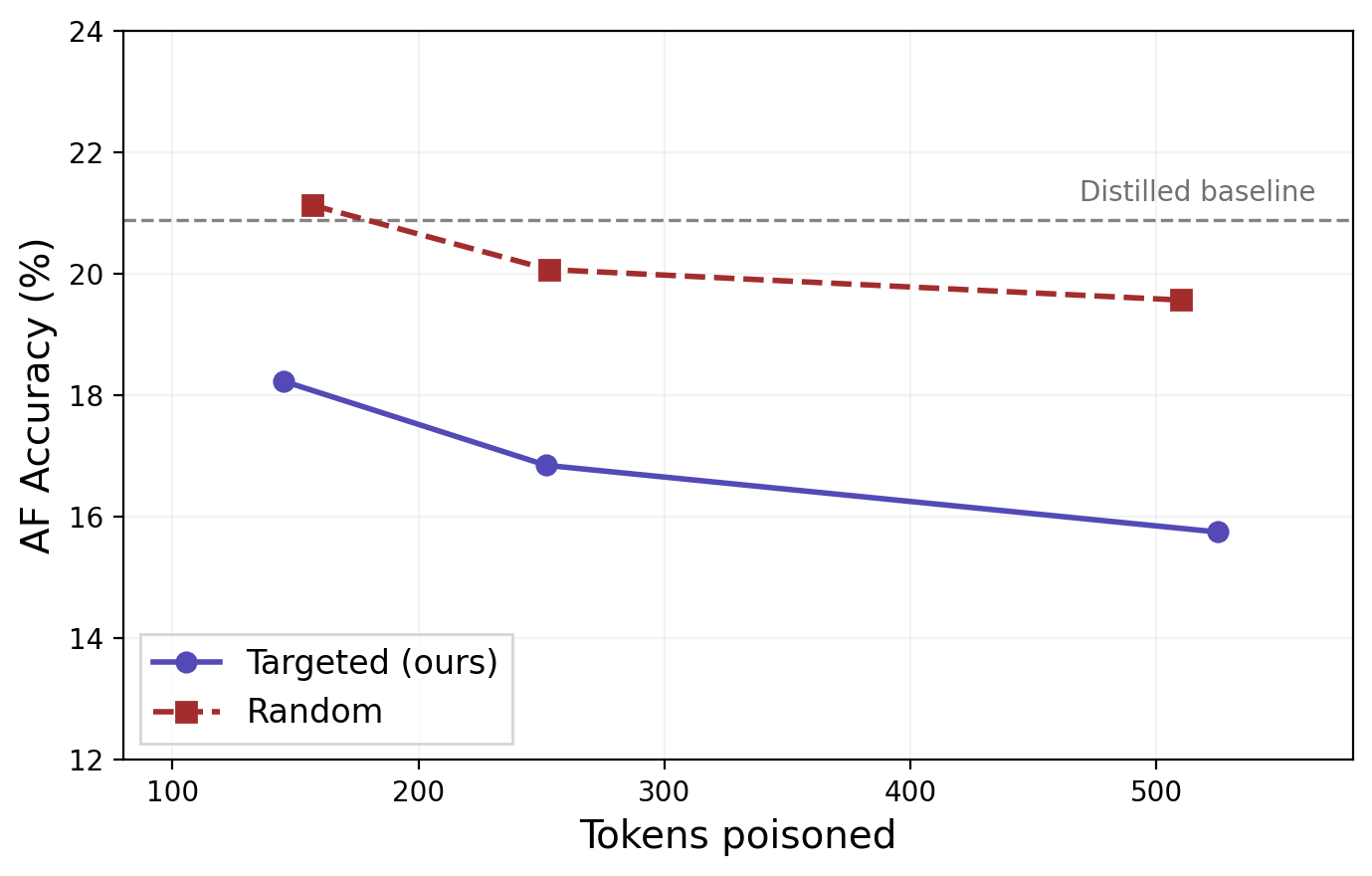}
    \caption{Accuracy drop between baseline versus poisoned reasoning trace distillation for varying numbers of tokens based on \textbf{TARS versus Random Sentence Removal.} Based on the figure, poisoning random sentences has a substantially lower impact on distillation performance than the branching sentences. In this experiment, the teacher model is DeepSeek R1 Distill Qwen 7B and the student model is Llama 3.2 3B.}
    \label{fig:random_versus_branching_drop}
\end{figure}

From Figure \ref{fig:random_versus_branching_drop}, the distillation performance gap is minimal for a varying number of tokens poisoned.

\subsection{Best-of-$n$ re-query attacks}\label{appendix:attacks}

We evaluate both defenses under an attacker who queries the teacher
$n{=}4$ times per problem and must commit to a single response before
distilling. All runs share a DeepSeek-R1-Distill-Qwen-7B teacher, a
Gemma-3-1B-IT student, MATH ($2{,}470$ train / $2{,}480$ test problems),
$\tau{=}0.6$, and a generation budget of $2048$ tokens. TARS deletes $k{=}10$ sentences per trace and to be comparable, antidistillation sampling
(ADS) uses $\lambda{=}0.11$, $\varepsilon{=}0.01$ with a Qwen2.5-3B proxy
student. All students are trained with a single SFT seed (42).

\paragraph{Attack A: correctness then length.}
The attacker discards incorrect responses and, among those remaining,
keeps the \emph{longest}; if no response is correct, the longest response
overall is kept. This thus, specifically, targets TARS, which damages
a trace by deleting sentences: the longest re-query is likely the least damaged
one. The same rule is applied to the ADS pool for comparability, even
though it is not adapted to how ADS degrades a trace---thus giving ADS an advantage.

\begin{table}[H]
\centering
\caption{Attack A (correctness, then longest response). Teacher accuracy is
measured on the training traces the student is distilled from; student
accuracy is measured on MATH test. Controls distill on a single un-selected
defended draw and are shared with Table~\ref{tab:attack-b}, since they do not
depend on which attack is applied. Positive $\Delta$ means the attack
succeeded, i.e.\ the defense was weakened.}
\label{tab:attack-a}
\begin{tabular}{llcccc}
\toprule
& & Teacher & \multicolumn{2}{c}{Student accuracy} & \\
\cmidrule(lr){4-5}
Defense & Attacker & accuracy & Raw & Answer-forced \\
\midrule
\multirow{2}{*}{Thought-anchor removal ($k{=}10$)}
  & none (control) & 0.556 & 0.0980 & 0.1185\\
  & longest        & 0.686 & 0.1012 & \textbf{0.1194} \\
\addlinespace
\multirow{2}{*}{Antidistillation sampling ($\lambda{=}0.11$)}
  & none (control) & 0.366 & 0.0577 & 0.0794\\
  & longest        & 0.718 & 0.0835 & \textbf{0.1097} \\
\bottomrule
\end{tabular}
\end{table}

\paragraph{Attack B: defense-specific selection with repetition filtering.}
Attack A applies one lever to both defenses, which understates the
attacker against ADS: ADS does not shorten traces, it drives them into
repetition until they exhaust the generation budget. Attack B therefore
gives each defense a matched lever. Against thought-anchor removal the attacker keeps the \emph{longest}
correct response, stepping to progressively shorter ones only when the
longest is truncated or predominantly repetitive; against ADS it keeps
the \emph{shortest}, stepping to progressively longer ones under the same
condition. It then strips consecutive repetition tokens from the
committed response. A repeated span must span at least $16$ redundant tokens to be removed, which leaves
ordinary notation such as \texttt{1005} intact, and a strip that would
invalidate the answer is reverted. Stripping removes $239$ tokens per
trace on average under ADS and essentially none under thought-anchor
removal, with no loss of teacher-side accuracy in either case.

\begin{table}[H]
\centering
\caption{Attack B (defense-specific selection with repetition filtering).
Against thought-anchor removal the attacker keeps the longest correct
response; against ADS the shortest. In both cases it then steps down that
ordering past truncated or predominantly repetitive candidates and strips
consecutive repetition tokens from the committed response.}
\label{tab:attack-b}
\begin{tabular}{llcccc}
\toprule
& & Teacher & \multicolumn{2}{c}{Student accuracy} & \\
\cmidrule(lr){4-5}
Defense & Attacker & accuracy & Raw & Answer-forced \\
\midrule
\multirow{2}{*}{Thought-anchor removal ($k{=}10$)}
  & none (control)        & 0.556 & 0.0980 & \textbf{0.1185}\\
  & longest $+$ filtering & 0.686 & 0.0782 & 0.0903 \\
\addlinespace
\multirow{2}{*}{Antidistillation sampling ($\lambda{=}0.11$)}
  & none (control)         & 0.366 & 0.0577 & 0.0794 \\
  & shortest $+$ filtering & 0.718 & 0.0919 & \textbf{0.1052} \\
\bottomrule
\end{tabular}
\end{table}

\paragraph{Discussion.}
Against ADS both attacks succeed: Attack A recovers $3.0$ points of
answer-forced accuracy over the unattacked control and Attack B recovers
$2.6$ points ($+2.6$ and $+3.4$ points respectively on raw accuracy).

Against thought-anchor removal neither attack helps. Attack A is
indistinguishable from the control ($+0.0009$), and Attack B is
\emph{worse} ($-0.0282$).

Taking each defense's best attacker outcome, thought-anchor removal admits
$0.1194$ and ADS admits $0.1097$, so thought-anchor removal is the weaker
defense by roughly $1$ point under a best-responding attacker, and both
attacks leave the ADS-distilled student below the thought-anchor one. However, referring to Figure~\ref{fig:teacher-student-accuracy-math}, we can see that TARS with k=10 does perform marginally worse than ADS with $\tau=0.11$ so such a difference in student performance is expected. The main takeaway is the robustness of TARS to attacks versus the fragility of ADS to simple post-processing schemes.

\subsection{What each defense does to a trace}\label{appendix:trace_stats}

To separate \emph{how} the two defenses damage a reasoning trace, we
characterise in Table~\ref{tab:trace-stats} three populations of teacher generations from
($2{,}470$ MATH problems: the undefended trace,
the same trace after thought-anchor deletion ($k{=}10$), and an
independently generated ADS trace at the matched seed
($\lambda{=}0.11$, $\varepsilon{=}0.01$).
\begin{table}[t]
\centering
\caption{Trace statistics for the three populations.
Thought-anchor removal damages a trace by making it shorter; ADS damages
it by filling it with repetition.}
\label{tab:trace-stats}
\begin{tabular}{lccc}
\toprule
 & Undefended & Thought-anchor & ADS \\
Statistic & (original) & removal ($k{=}10$) & ($\lambda{=}0.11$) \\
\midrule
\multicolumn{4}{l}{\emph{Length}} \\
Mean tokens                     & 1307 & 1024 & \textbf{1535} \\
Truncated at budget (\%)        & 50.9 & 58.0 & \textbf{87.7} \\
Sentences deleted               & ---  & 9.86 & --- \\
\midrule
\multicolumn{4}{l}{\emph{Content vs.\ repetition}} \\
Repetition tokens               & 0.6  & 0.4  & \textbf{628.0} \\
Repetition tokens (\% of trace) & 0.0  & 0.0  & \textbf{40.9} \\
\midrule
\multicolumn{4}{l}{\emph{Degeneration}} \\
Degenerate traces (\%)          & 1.0  & 0.7  & \textbf{59.5} \\
Traces with a tail loop (\%)    & 0.1  & 0.1  & \textbf{51.0} \\
Longest single-token run        & 1.7  & 1.7  & \textbf{281.6} \\
Compression ratio               & 0.382 & 0.509 & \textbf{0.220} \\
\midrule
\multicolumn{4}{l}{\emph{Accuracy}} \\
Trace accuracy                     & 0.560 & 0.560 & \textbf{0.363} \\
Absolute drop vs.\ undefended      & ---   & 0.000 & \textbf{0.198} \\
Relative drop vs.\ undefended (\%) & ---   & 0.0   & \textbf{35.3} \\
\bottomrule
\end{tabular}
\end{table}

\paragraph{Statistics tracked.}
\emph{Length} statistics are token counts under the teacher's tokenizer.
A trace is \emph{truncated} if it reached the $2048$-token generation
budget without emitting an end-of-sequence token, meaning it never
produced a final answer. \emph{Repetition tokens} counts positions
belonging to a block of some period $p \le 32$ that repeats back to back
($aaaa$, $ababab$, and so on); the first pass through a block is never
counted, and a span must contain at least $16$ redundant tokens, so
ordinary notation such as \texttt{1005} is not counted as repetition.
Under \emph{degeneration}, a trace is \emph{degenerate} if a single token repeats at least $10$
times consecutively or if more than half its 4-grams are repeats; a
\emph{tail loop} is a repeating cycle detected at the end of the trace;
the \emph{longest single-token run} is the maximum number of consecutive
identical tokens; and the \emph{compression ratio} is zlib-compressed size over raw size, where
lower means more redundant. \emph{Trace accuracy} is the fraction of traces whose final answer
verifies against the reference solution.

\paragraph{Trace-level observations.}
Thought-anchor removal behaves exactly as designed: it strips $283$
tokens ($21.7\%$) and $9.86$ sentences per trace while leaving accuracy
unchanged at $0.560$, because deletion spares the final answer. Its
repetition profile is indistinguishable from the undefended trace, at
$0.4$ against $0.6$ repetition tokens.

ADS damages traces by a different mechanism entirely. Its traces are
$1.17\times$ \emph{longer} than undefended ones, yet $40.9\%$ of their
tokens are consecutive repetition so ADS yields the least real content of the three despite producing the
most text. Degeneration is pervasive throughout ADS generations:
$59.5\%$ of ADS traces are degenerate against $1.0\%$ undefended, half
contain a tail loop, and the mean longest single-token run is $281.6$
tokens against $1.7$. Consistent with this, $87.7\%$ of ADS traces
exhaust the generation budget without emitting an answer, and trace
accuracy falls $19.8$ points.



\section{Distillation experimental setup}
\label{appendix:experimental_setup}

In this work, we run knowledge distillation experiments on the student models using traces generated by the teacher model. The models are trained using HuggingFace's SFTTrainer with the following optimal hyperparameters:

\begin{table}[H]
\centering
\label{tab:hyperparameters}
\begin{tabular}{@{}ll@{}}
\toprule
\textbf{Parameter} & \textbf{Value} \\ \midrule
Optimizer & Fused AdamW \\
Learning Rate & $5 \times 10^{-4}$ \\
LR Schedule & Cosine decay with 3\% linear warmup \\
Batch Size & 16 (via gradient accumulation) \\
Epochs & 3 \\
Weight Decay & 0.1 \\
Gradient Clipping & 1.0 (max norm) \\
Context Window & 4096 tokens \\
Precision & \texttt{bfloat16} (mixed precision) \\
Attention Mechanism & Flash Attention (\texttt{sdpa}) \\ \bottomrule
\end{tabular}
\caption{Hyperparameters and Hardware Settings for Student Model Distillation}
\end{table}

These hyperparameters are very similar to those used in the Antidistillation Sampling method \citep{SAVANITFX2025}.

When training, we used 4 NVIDIA A100s. 

\section{Proofs}\label{appendix:proofs}
\begin{theorem}[Detectability of Burst Deletions; Formal version of Proposition \ref{prop:delete}]\label{thm:deletebound}
    Let $X = (X_1, \ldots, X_N)$ be generated by a time-homogeneous $m$-order Markov source. Thus, for every $t$
    \begin{equation}
    P(X_t = a \mid X_{<t} = x_{<t}) = p(a \mid x_{t-m:t-1}),
    \end{equation}
    where a fixed prompt is used for $t\leq m$. Assume there is a common finite support $\mathcal{A}$ and a constance $\gamma \in (0,1)$ such that
    \begin{equation}\label{eq:positive}
    p(a\mid h) \geq \gamma >0, 
    \end{equation}
    for all $m$-length histories $h$ and for every $a \in \mathcal{A}$. For a deletion mask $s \subset [N] := \{1,2,\ldots N\}$, let 
    \begin{equation}
    [N]\backslash s = \{r_1, ..., r_{M_s}\},
    \end{equation}
    with $r_1 < \ldots < r_{M_s}$ be the retained positions and define
    \begin{equation}
    T_s(x) := (x_{r_!}, \ldots, x_{r_{M_s}}).
    \end{equation}
    Let $b(s)$ be the number of maximal contiguous intervals in s. Thus, $b(s)$ is the number of deletion bursts (i.e. $k$ in our context). Let $P_s$ be the distribution of $T_s(X)$ and let $P_0^{(M_s)}$ be the distribution of a clean length-$M_s$ sample from the same Markov source. Then, for every fixed mask s, 
    \begin{equation}
    D_{KL}\left(P_s \parallel P_0^{(M_s)}\right) \leq b(s)m\log\left(\frac{1}{\gamma}\right) = km\log\left(\frac1\gamma\right).
    \end{equation}
\end{theorem}
\begin{proof}
    Fix a mask $s$ and lets its maximal contiguous deleted intervals be 
    \begin{equation}
    s = I_1 \cup \cdots \cup I_{b(s)},
    \end{equation}
    where $I_j = \{u_j, u_j+1, \ldots, v_j\}$. Let the retained source positions be $r_1 < \cdots < r_{M_s}$. Then, the $t$-th output token is $Y_t = X_{r_t}$. Define \begin{equation}
    A_s := \{t \in [M_s] : \{r_t - m, \ldots, r_t-1\} \cap s \neq \emptyset\}.
    \end{equation}
    We can see that $t\in A_s$ when at least one deleted token lies in the original $m$-token context immediately preceding $X_{r_t}$. We will now use the following lemma to bound $|A_s|$.
    \begin{claim}\label{claim:Asbound}
        $|A_s| \leq mb(s)$.
    \end{claim}
    \begin{proof}
        Consider one deletion burst $I_j = [u_j,v_j]$. Suppose this burst affects an output position $t$. This implies that \begin{equation}\label{eq:deleteintersect}
        [r_t-m, r_t-1]\cap [u_j,v_j] \neq \emptyset.
        \end{equation}
        Since $r_t$ is retained by construction, it is not in the deletion interval and since the interval intersects the past of $r_t$, it must entirely lie before $r_t$. Thus, 
        \begin{equation}
            v_j < r_t
        \end{equation}
         and by \ref{eq:deleteintersect}, 
         \begin{equation}
             r_t-m \leq v_j 
         \end{equation}
         so 
         \begin{equation}
             v_j < r_t \leq v_j+m.
         \end{equation}
         Thus, burst $I_j$ can only affect the $m$ tokens $v_j+1, \ldots, v_j+m$ so each deletion burst can only affect $m$ retained-token positions. Taking a union over all bursts gives
        \begin{equation}
            |A_s| \leq mb(s)
        \end{equation}
    \end{proof}
    \begin{remark}
        The length of the burst does not appear in this bound at all.
    \end{remark}
    Now, define 
    \begin{equation}
        q_t^s(a\mid y_{<t}) := P_s(Y_t=a\mid Y_{<t} = y_{<t}).
    \end{equation}
    We want to show that $q_t^s$ does not contribute to the KL-divergence (namely the likelihood ratio with respect to the clean source) for $t\notin A_s$.
    \begin{claim}\label{claim:conddist}
        $q_t^s(a\mid y_{<t}) = p_t^0(a\mid y_{<t})$ for $t\notin A_s$ where $p_t^0$ is a clean next-token distribution conditioned on the previous tokens being $y_{<t}$. Thus, only output positions in $A_s$ can contribute to the likelihood ratio.
    \end{claim}
    \begin{proof}
        Because $Y_t = X_{r_t}$, 
        \begin{align}
            q_t^s(a\mid y_{<t}) &= P(X_{r_t} = a \mid X_{r_1} = y_1, \ldots, X_{r_{t-1}}=y_{t-1})\\
            &= \mathbb{E}\left[ \mathbbm{1}\{ X_{r_t}=a\}  \mid Y_{<t} = y_{<t}\right]\\
            &= \mathbb{E}[\mathbb{E}[\mathbbm{1}\{ X_{r_t}=a \} \mid X_{<r_t} = x_{<r_t}] \mid Y_{<t} = y_{<t}]\\
            &\stackrel{(a)}{=} \mathbb{E}[p(a\mid x_{<r_t}) \mid Y_{<t}=y_{<t}]\\
            &\stackrel{(b)}{=} \mathbb{E}[p(a\mid x_{r_t-m:r_t-1}) \mid Y_{<t}=y_{<t}],
        \end{align}
        where $(a)$ follows by the tower property and $(b)$ follows by the Markov source assumption. By assumption, $p(a\mid x_{t-m:t_1}) \geq \gamma$ so $q_t^s(a\mid y_{<t}) \geq \gamma$. Now, if $t\notin A_s$, then none of the source positions
        \[
        r_t-m, \ldots, r_t-1,
        \]
        was deleted. Thus,
        \[
        X_{r_t-m:r_t-1} = y_{t-m:t-1},
        \]
        so 
        \begin{equation}
            q_t^s(a\mid y_{<t}) = p(a\mid y_{t-m:t-1}) \stackrel{(c)}{=} p_t^0(a\mid y_{<t}),
        \end{equation} for $t\notin A_s$, where $(c)$ follows by the Markov assumption. 
        For $t\leq m$ and $t\notin A_s$, if $r_t > t$, $[\max\{1,r_t-m\}, r_t-1]$ contains at least $t$ source positions. Since $t\notin A_s$, all of these positions are retained which contradicts the fact that only $t-1$ retained positions precede $r_t$. Thus $r_t = t$. Thus, $q_t^s(a\mid y_{<t}) = p_t^0(a\mid y_{<t})$ as above.
    \end{proof}
    We can see that for an affected position $t\in A_s$, $q_t^s(a\mid y_{<t})\geq \gamma$ and $p_t^0(a\mid y_{<t}) \geq \gamma$. Thus, since $\gamma \in (0,1)$,
    \begin{equation}
        \gamma \leq \frac{q_t^s(a\mid y_{<t})}{p_t^0(a\mid y_{<t})}\leq \frac1\gamma.
    \end{equation}
    Thus, 
    \begin{equation}
        \left|\log \frac{q_t^s(a\mid y_{<t})}{p_t^0(a\mid y_t)} \right| \leq \log \frac1\gamma.
    \end{equation}
    Now we have that,
    \begin{align}
        \log \frac{P_s(y)}{P_0^{(M_s)}(y)} &= \log \frac{\prod_{t=1}^{M_s} q_t^s(y_t \mid y_{<t})}{\prod_{t=1}^{M_s}p_t^0(y_t\mid y_{<t})}\\
        &= \sum_{t=1}^{M_s} \log \frac{q_t^s(y_t\mid y_{<t})}{p_t^0(y_t\mid y_{<t})}.
    \end{align}
    By Claim \ref{claim:conddist}, this simplifies to
    \begin{align}
        \left| \log \frac{P_s(y)}{P_0^{(M_s)}(y)} \right| &= \left| \sum_{t\in A_s} \log \frac{q_t^s(y_t\mid y_{<t})}{p_t^0(y_t\mid y_{<t})} \right|\\
        &\leq \sum_{t\in A_s} \left|  \log \frac{q_t^s(y_t\mid y_{<t})}{p_t^0(y_t\mid y_{<t})} \right|\\
        &\leq |A_s|\log \frac1\gamma\\
        &\stackrel{(d)}{\leq} b(s)m\log \frac1\gamma = km\log\frac1\gamma,\label{eq:pointwise}
    \end{align}
    where $(d)$ follows by Claim \ref{claim:Asbound}.
    Now, taking an expectation over $Y\sim P_s$ we finally arrive at
    \begin{equation}
        D_{KL}(P_s\parallel P_0^{(M_s)}) \leq b(s)m\log\frac1\gamma = km\log\frac1\gamma.
    \end{equation}
\end{proof}
\begin{corollary}
    Under the same assumptions as Theorem \ref{thm:deletebound}, let the mask $S$ be selected according to 
    \begin{equation}
        q(s\mid x) = P(S=s \mid X=x).
    \end{equation}
    Suppose $\mathbb{E}[b(S)] = k$ and let 
    \begin{equation}
        Y = T_S(X),\quad L = |Y|,
    \end{equation}
    and let 
    \begin{equation}
        \nu(\ell) = P(L=\ell)
    \end{equation}
    be the output length after deletion. Define the length-matched reference distribution to be
    \begin{equation}
        \overline{P_0}(y) = \nu(|y|)P_0^{(|y|)}(y).
    \end{equation}
    Then 
    \begin{equation}
        D_{KL}(P_Y\parallel \overline{P_0}) \leq km\log\frac1\gamma + I(X;S).
    \end{equation}
    Namely,
    \begin{equation}
        \mathbb{E}_{L\sim\nu}\left[D_{KL}\left(P_{Y\mid L)}\parallel \overline{P}_0^{(L)}\right)\right] \leq km\log\frac1\gamma + I(X;S),
    \end{equation} and if all masks produce the exact same output length $M$ then
    \begin{equation}
        D_{KL}\left(P_Y\parallel P_0^{(M)}\right) \leq km\log\frac1\gamma + I(X;S)
    \end{equation}
\end{corollary}
\begin{remark}
    When $S=s$ a.s., $I(X;S) = 0$ and Theorem \ref{thm:deletebound} is recovered.
\end{remark}
\begin{proof}
    Let 
    \begin{equation}
        Q_{X,S}(x,s) = P_X(x)q(s\mid x).
    \end{equation}
    and 
    \begin{equation}
        \pi(s) := Q_S(s) = \sum_x P_X(x)q(s\mid x)
    \end{equation}
    be the marginal distribution of the masks. Now construct the independent coupling of $Q_{X,S}$ denoted
    \begin{equation}
        R_{X,S}(s,x) = P_X(x)\pi(s).
    \end{equation}
    We can then notice that
    \begin{equation}
        D_{KL}(Q_{X,S}\parallel R_{X,S}) = I_Q(X;S).
    \end{equation}
    We can now apply the deterministic deletion mapping
    \begin{equation}
        F(x,s) = (s,T_s(x)),
    \end{equation}
    and denote $Q_{S,Y}$ and $R_{S,Y}$ the respective pushforward measures under $F$. Notice that under $R$, the mask is independent of $X$ so conditional on $S = s$, $R_{S,Y}$ factors through $P_s$, namely
    \begin{equation}
        R_{S,Y}(s,y) = \pi(s)P_s(y).
    \end{equation}
    Thus, by the data-processing inequality
    \begin{equation}
        D_{KL}(Q_{S,Y}\parallel R_{S,Y}) \leq D_{KL}(Q_{X,S}\parallel R_{X,S}) = I_Q(X;S).
    \end{equation}
    We can now define the joint distribution of the clean reference distribution
    \begin{equation}
        B_{S,Y}(s,y) = \pi(s)P_0^{(M_s)}(y).
    \end{equation}
    We now prove the following claim in order to invoke the chain rule of the Radon Nikodym derivative.
    \begin{claim}\label{claim:abs_cont}
        $Q_{S,Y}\ll R_{S,Y} \ll B_{S,Y}$ (where $\ll$ denotes absolute continuity)
    \end{claim}
    \begin{proof}
        ($Q\ll R$) This absolute continuity comes for free because $R$ is the independent coupling of $Q$.\\
        ($R\ll B$) This follows because of Assumption \ref{eq:positive}. Namely, $R_{S,Y} = \pi(s)P_s(y) > 0$ means that $\pi(s) > 0$ and since $P_0^{(M_s)}(y) > 0$ by assumption \ref{eq:positive}, $B_{S,Y} = \pi(s)P_0^{(M_s)} > 0$ so $R\ll B$.
    \end{proof}
    Since, by Claim \ref{claim:abs_cont}, $Q\ll R \ll B$, we can invoke the invoke the chain rule of the Radon-Nikodym derivative to arrive at
    \begin{equation}
        D_{KL}\left(Q_{S,Y}\parallel B_{S,Y}\right) = D_{KL}(Q_{S,Y}\parallel R_{S,Y}) + \mathbb{E}_{Q_{S,Y}}\left[ \log \frac{R_{S,Y}(S,Y)}{B_{S,Y}(S,Y)} \right].
    \end{equation}
    We have that 
    \begin{align}
        \log\frac{R_{S,Y}(s,y)}{B_{S,Y}(s,y)} &= \log\frac{\pi(s)P_s(y)}{\pi(s)P_0^{(M_s)}(y)}\\
        &= \log\frac{P_s(y)}{P_0^{(M_s)}(y)}\\
        &\stackrel{(a)}{\leq} b(s)m\log \frac1\gamma,
    \end{align}
    where $(a)$ follows from the pointwise version of Theorem \ref{thm:deletebound} (namely, equation \ref{eq:pointwise}). Combining everything together, we get
    \begin{equation}
        D_{KL}(Q_{S,Y}\parallel B_{S,Y}) \leq I_Q(X;S) + k m\log \frac1\gamma,
    \end{equation} since $\mathbb{E}[b(S)] = k$.
    We know by the chain rule of KL-divergence that 
    \begin{equation}
        D_{KL}(Q_Y\parallel B_Y) \leq D_{KL}(Q_{S,Y}\parallel B_{S,Y}).
    \end{equation}
    Thus, we now need to identify what $B_Y$ is. Observe that for a sequence $y$ of length $\ell$,
    \begin{equation}
        B_Y(y) = \sum_{s:M_s = \ell} \pi(s)P_0^{(\ell)}(y) = P_0^{(\ell)}\sum_{s:M_s = \ell} \pi(s).
    \end{equation}
    However, observe that 
    \begin{equation}
        \sum_{s:M_s=\ell}\pi(s) = P(M_s = \ell) = \nu(\ell),
    \end{equation}
    so $B_Y(y) = \nu(|y|)P_0^{(|y|)}(y) = \overline{P_0}(y).$ Finally, combining everything,
    \begin{equation}
        D_{KL}(P_Y\parallel \overline{P_0})\leq km\log \frac1\gamma + I(X;S).
    \end{equation}
\end{proof}
\begin{remark}
    Since $P_Y$ and $\overline{P_0}$ have the same length marginal $\nu$, we have 
    \begin{equation}
        D_{KL}(P_Y\parallel \overline{P_0}) = \mathbb{E}_{L\sim \nu}\left[D_{KL}\left( P_{Y\mid L} \parallel P_0^{(L)} \right) \right] \leq km\log \frac1\gamma + I(X;S).
    \end{equation}
    Similarly, if there is always a fixed output length $M$ we have 
    \begin{equation}
        D_{KL}\left( P_Y\parallel P_0^{(M)} \right) \leq km\log \frac1\gamma + I(X;S)
    \end{equation}
\end{remark}

\end{document}